\documentclass[a4paper,12pt]{article}

\usepackage[utf8]{inputenc}

\usepackage{geometry}

\usepackage{pgf}
\usepackage{nicefrac}
\usepackage{booktabs}
\usepackage{enumitem}
\usepackage{bm}
\usepackage{fullpage}
\usepackage{tikz}
\usetikzlibrary{automata,arrows,calc,chains,decorations.pathreplacing,decorations.text,decorations.pathmorphing,external}

\pgfdeclarelayer{edgelayer}
\pgfdeclarelayer{nodelayer}
\pgfsetlayers{edgelayer,nodelayer,main}

\tikzset{vertex/.style={draw,shading=ball,ball color=black,circle, inner sep=0pt, minimum size=5pt}}
\tikzset{vertex2/.style={draw,shading=ball,ball color=black!25,circle,  inner sep=0pt, minimum size=5pt}}

\makeatletter
\tikzset{
    dot diameter/.store in=\dot@diameter,
    dot diameter=1.25pt,
    dot spacing/.store in=\dot@spacing,
    dot spacing=4pt,
    dots/.style={
        line width=\dot@diameter,
        line cap=round,
        dash pattern=on 0pt off \dot@spacing
    }
}
\makeatother

\usepackage[]{asymptote}
 \usepackage{pgfplots}

\usepackage[%
]{graphicx}

\usepackage[
   centertags, 
   sumlimits,  
   intlimits,  
   namelimits, 
   fleqn,     
]{amsmath} %

\usepackage{amsthm}
\usepackage{amssymb}
\usepackage{mathtools}

\usepackage[T1]{fontenc}
\usepackage{lmodern}
\usepackage[%
	english
]{nomencl}

\definecolor{sectioncolor}{RGB}{0, 0, 0}    
\definecolor{textcolor}{RGB}{0, 0, 0}        
\definecolor{shadecolor}{gray}{0.90}
\definecolor{pdfurlcolor}{rgb}{0,0,0.6}
\definecolor{pdffilecolor}{rgb}{0.7,0,0}
\definecolor{pdflinkcolor}{rgb}{0,0,0.6}
\definecolor{pdfcitecolor}{rgb}{0,0,0.6}
\colorlet{stringcolor}{green!40!black!100}
\colorlet{commencolor}{blue!0!black!100}

\usepackage{aliascnt}

\usepackage[natbib=true,
	    bibencoding=utf8,
	    sorting=nyt,
	    style=alphabetic,
	    useprefix=true,
            hyperref=true,
            backref=false,
	    backrefstyle=none,
	    alldates=short,
            abbreviate=true,
	    isbn=false,
	    doi=false,
	    firstinits=true,
	    maxbibnames=99,
	    backend=biber]{biblatex}
	    
\DeclareBibliographyCategory{needsurl}

\renewbibmacro*{url+urldate}{%
  \ifcategory{needsurl}{
    \printfield{url}%
    \iffieldundef{urlyear}
      {}
      {\setunit*{\addspace}%
       \printurldate}}
    {}}
    
\AtEveryBibitem{%
  \clearlist{language}%
  \clearlist{publisher}%
  \clearlist{location}%
  \clearname{editor}%
}

\renewbibmacro{in:}{}

\theoremstyle{plain}
 
\newtheorem{theorem}{Theorem}  
 
\newaliascnt{lemma}{theorem}  
\newtheorem{lemma}[lemma]{Lemma}  
\aliascntresetthe{lemma}

\newtheorem*{theorem*}{Theorem}

\newtheorem*{corollary}{Corollary}

\theoremstyle{definition}
\newtheorem{definition}{Definition}[section]

\theoremstyle{remark}

\usepackage[
  colorlinks=true,         
  urlcolor=pdfurlcolor,    
  filecolor=pdffilecolor,  
  linkcolor=pdflinkcolor,  
  citecolor=pdfcitecolor,  %
]{hyperref}

\usepackage{csquotes}
\usepackage[
	english,
]{babel}
\renewcommand{\leq}{\leqslant}
\renewcommand{\geq}{\geqslant}

\let\epsilon\varepsilon

\begin{document}

\title{\bf Approximation Complexity of Max-Cut on Power Law Graphs\\[1ex]}
\author{Mikael Gast\thanks{Dept. of Computer Science and the Hausdorff Center for Mathematics, University of Bonn.
    e-mail:{ \texttt{\href{mailto:gast@cs.uni-bonn.de}{gast@cs.uni-bonn.de}}}} \and
	Mathias Hauptmann\thanks{Dept. of Computer Science, University of Bonn.
    e-mail:{ \texttt{\href{mailto:hauptman@cs.uni-bonn.de}{hauptman@cs.uni-bonn.de}}}} \and
	Marek Karpinski\thanks{Dept. of Computer Science and the Hausdorff Center for Mathematics, University of Bonn. Research supported by DFG grants and the Hausdorff grant EXC59-1/2.
    e-mail:{ \texttt{\href{mailto:marek@cs.uni-bonn.de}{marek@cs.uni-bonn.de}}}}}
\date{}
\maketitle

\begin{abstract}
 In this paper we study the MAX-CUT 
 problem on power law graphs (PLGs) with power law exponent $\beta$. 
We prove some new approximability results on that problem.
In particular we show that there exist polynomial time approximation schemes (PTAS) for MAX-CUT on PLGs for the power law exponent $\beta$ in the interval $(0,2)$. 
For $\beta>2$ we show that for some $\epsilon>0$, MAX-CUT is NP-hard to approximate within approximation ratio $1+\epsilon$, 
ruling out the existence of a PTAS in this case. 
Moreover we give an approximation algorithm with improved constant approximation ratio for the case of $\beta>2$. 
\end{abstract}



\section{Introduction}

In the study of large-scale complex networks, a large body of theoretical and practical work is devoted to clustering and partitioning problems \cite{Newman2004,Wu2004,Newman2006a,Dinh2013}.
The aim is to identify and to characterize natural partition structures of existing real world networks, like protein interaction networks, online social networks and parts or layers of the World Wide Web.

Given an undirected graph $G=(V,E)$ the MAX-CUT problem asks for a partition of the vertex set into two parts such as to maximize the number of edges between them.
MAX-CUT is one of the classic $21$ NP-complete problems listed in \cite{Karp1972} and has applications ranging from VLSI design and cluster analysis to statistical physics \cite{Barahona1988}.

In this paper we consider the MAX-CUT problem in the context of large-scale complex networks, more precisely in the context of so called \emph{power law graphs} (PLG).
The distinguishing feature of power law graphs is that their node degree distribution follows a \emph{power law}, that is, the number of nodes of degree $i$ is proportional to $i^{-\beta}$, for some fixed \emph{power law exponent} $\beta >0$.
A power law degree distribution has been observed for a large number and variety of social, information, technological and biological networks \cite{Clauset2009a}.

\paragraph{Previous Results}

In the general setting partition problems like the MAX-CUT (and MAX-BISECTION) problem were shown to be APX-hard.
The greedy algorithm or random assignment algorithm for MAX-CUT was shown to yield a $\nicefrac{1}{2}$-approximation for the problem \cite{Sahni1976}.
Beside some $(\nicefrac{1}{2}+o(1))$-approximation algorithms, the currently best approximation ratio 
for the problem is due to \citeauthor{Goemans1995} in their seminal paper \cite{Goemans1995}.
Using \emph{semidefinite programming} they achieved an $\alpha_{GW}$-approximation algorithm for MAX-CUT, where $\alpha_{GW}\approx 0.879$ is the trigonometric quantity $\alpha_{GW}=\nicefrac{2}{\pi \sin\varTheta}$ for $\varTheta\approx 2.33$.

The APX-hardness of the problem on general instances was shown by \citeauthor{Papadimitriou1991} \cite{Papadimitriou1991}.
Under the assumption $\text{P}\neq \text{NP}$, the first explicit inapproximability bound of $\nicefrac{84}{83}$ was proven by \citeauthor{Bellare1998} \cite{Bellare1998} and later refined to the current best bound of $\nicefrac{17}{16}$ due to \citeauthor{Hastad2001} \cite{Hastad2001}  and \citeauthor{Trevisan2006} \cite{Trevisan2006}.
Assuming that the Unique Games Conjecture (UGC) holds true, \citeauthor{Khot2007} \cite{Khot2007} showed that MAX-CUT is UGC-hard to approximate to within $\alpha_{GW}+\epsilon$.
Thus, assuming the UGC, $\alpha_{GW}$ is the best possible approximation ratio achievable in the general case. 

Furthermore, a lot of interest has been devoted to the case of more restricted graph instances of MAX-CUT.
One of the first results was a polynomial time algorithm for the case of planar graphs \cite{Hadlock1975}. 
In the \emph{metric} variant of MAX-CUT (and MIN-BISECTION), instances are complete graphs where edge  weights are given by the metric distance of the corresponding endpoints.
For general metrics a randomized PTAS for MAX-CUT is due to \citeauthor{FernandezdelaVega2001} \cite{FernandezdelaVega2001} and for the case of tree metrics, a polynomial time algorithm was constructed by \citeauthor{Karpinski2013c} \cite{Karpinski2013c}.
A PTAS for metric MIN-BISECTION was shown by \citeauthor{FernandezdelaVega2004} \cite{FernandezdelaVega2004}, whereas for the general case the existence of a PTAS remains an open question (cf. \cite{Karpinski2002}).

Another important special case is when the corresponding problem instances are \emph{dense}, i.e. for some $\delta>0$, the number of edges is lower bounded by $\delta\cdot n^2$.
\citeauthor{Arora1995} \cite{Arora1995} gave a PTAS for dense MAX-CUT and, more generally, for dense MAX-$k$-CSPs.
The result was extended also to sub-dense instances in \cite{FernandezdelaVega2005,Bjorklund2005} and to dense weighted instances in \cite{FernandezdelaVega2000}.

Regarding lower approximation bounds, a lot of interest has been devoted to the case of MAX-CUT in degree $d$ bounded graphs and $d$-regular graphs.
In a series of papers \cite{Berman1999,Berman2001} \citeauthor{Berman2001} showed, among other results, that the MAX-CUT problem restricted to $3$-regular graphs is NP-hard to approximate to within $\nicefrac{152}{151}$.

\section{Main Results}
In this paper we study the MAX-CUT problem on power law graphs.
Our main results are new and improved upper approximation bounds for the problem.
In particular we show that there exists a polynomial time approximation scheme (PTAS) for MAX-CUT on $(\alpha,\beta)$-power law graphs for $0<\beta<2$.

For the range $\beta<1$, we observe that $(\alpha,\beta)$-power law graphs are {dense} (in the average sense, i.e. the number of edges is $\Omega(n^2)$) and the result of \citeauthor{Arora1995} \cite{Arora1995} can be applied to obtain a PTAS.
For $\beta=1$ the graphs are not dense anymore. For this case we prove that $(\alpha,\beta)$-power law graphs are \emph{core-dense} and use the result of \citeauthor{FernandezdelaVega2005} \cite{FernandezdelaVega2005} yielding a PTAS for the problem.

In the range $1<\beta<2$, none of the above results directly apply because instances are neither dense nor core-dense.
In order to construct a PTAS in this case, we partition the vertex set of power law MAX-CUT instances into two sets of high degree vertices and low degree vertices.
We show that for a suitable choice of the partition parameters the induced subgraph of high degree vertices is asymptotically dense, and at the same time the total number of edges induced by the low degree vertices is small.
Thus, in order to obtain a $\frac{1}{1+\epsilon}$-approximate cut, 
we run the algorithm of \cite{Arora1995} on the subgraph of high degree vertices and afterwards placing the remaining vertices arbitrary.

For $\beta>2$ we show that the Goemans-Williamson algorithm \cite{Goemans1995} can be combined with a preprocessing to yield an improved constant approximation ratio. 
Moreover we show that for $\beta>2$, MAX-CUT cannot be approximated with a constant approximation ratio arbitrary close to $1$. For this purpose we use the lower bound result of \cite{Berman1999,Berman2001}, construct an embedding of low degree graphs into power law graphs and obtain in this way also explicit approximation lower bounds depending on the power law exponent $\beta>2$. 
Moreover, a variant of this construction also proves the NP-hardness of MAX-CUT in power law graphs for the whole range $\beta>0$.

Besides NP-hardness in the exact setting, the status of MAX-CUT in PLGs for $\beta=2$ remains unsettled.
However, we consider the case when $\beta$ is a function of the size of the PLG that converges to $2$ from below.
We call this the \emph{functional case}.
In particular we show that for $\beta_f=2-\frac{1}{f(\alpha)}$, MAX-CUT in $(\alpha,\beta_f)$-PLGs admits a PTAS provided the convergence of $\beta_f$ to $2$ is sufficiently slow, namely for all sublinear functions $f(\alpha)=o(\alpha)$.

\paragraph{Organization of the Paper} 
\autoref{sec:Preliminaries} provides the definition of the $(\alpha,\beta)$-PLG model due to \cite{Aiello2001} and related notations.
In \autoref{sec:Approx} we present our PTAS constructions for MAX-CUT on $(\alpha,\beta)$-PLG for $0<\beta<2$.
The functional case $\beta_f=2-\frac{1}{f(\alpha)}$ is considered in \autoref{sec:functional}.
Furthermore we show  an improved constant approximation ratio for the case $\beta>2$ in \autoref{sec:beta>2}.
Finally, in \autoref{sec:LowerBounds}, we prove APX-hardness of the problem for the case $\beta>2$ and NP-hardness for the whole range $\beta>0$.


%
%
%

\section{Preliminaries}\label{sec:Preliminaries}

In this section we first give the formal definition of $(\alpha,\beta)$-power law graphs. Then we provide notations for 
sizes and volumes of some subsets of the vertex set of a power law graph which we call \emph{intervals}. Later on we will give estimates of these quantities in the analysis of our upper and lower bound constructions for MAX-CUT.
\begin{definition}\cite{Aiello2001}
An undirected multigraph $G=(V,E)$ with self loops is called an $(\alpha,\beta)$ power law graph if the following conditions hold:
\begin{itemize}
\item The maximum degree is $\Delta=\lfloor e^{\alpha\slash\beta}\rfloor$.
\item For $i=1,\ldots , \Delta$, the number $y_i$ of nodes of degree $i$ in $G$ satisfies
      \[y_i = \left\lfloor \frac{e^{\alpha}}{i^{\beta}}\right\rfloor\] 
\end{itemize}
\end{definition}
The following estimates for the number $n$ of vertices of an $(\alpha,\beta)$-power law graph are well known \cite{Aiello2001}:
\[n\approx \left\{\begin{array}{l@{\quad}l}
 \frac{e^{\alpha\slash\beta}}{1-\beta} & \mbox{for $0<\beta <1$,}\\
 \alpha\cdot e^{\alpha} & \mbox{for $\beta =1$,}\\
 \zeta (\beta )\cdot e^{\alpha} & \mbox{for $\beta >1$.}
\end{array}\right.\quad 
m\approx \left\{\begin{array}{l@{\quad}l}
 \frac{1}{2}\frac{e^{2\alpha\slash\beta}}{2-\beta} & \mbox{for $0<\beta <2$,}\\
 \frac{1}{4}\alpha e^{\alpha} & \mbox{for $\beta =2$,}\\
 \frac{1}{2}\zeta (\beta -1)e^{\alpha} & \mbox{for $\beta >2$.}
\end{array}\right.\]
Here $\zeta (\beta)=\sum_{i=1}^{\infty}i^{-\beta}$ is the {\sl Riemann Zeta Function}.

A random model for $(\alpha,\beta)$-power law graphs was given in \cite{Aiello2001} and is constructed in the following way:
\begin{enumerate}
  \item Generate a set $L$ of $\text{deg}_G(v)$ distinct copies of each vertex $v$.
  \item Generate a random matching on the elements of $L$.
  \item For each pair of vertices $u$ and $v$, the number of edges joining $u$ and $v$ in $G$ is equal to the number of edges in the matching of $L$, which join copies of $u$ to copies of $v$.
\end{enumerate}

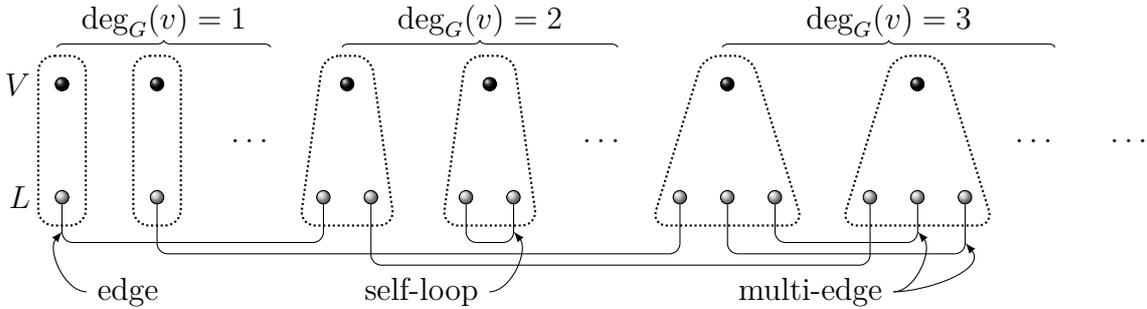
\begin{figure}[htb]
\centering
 \begin{tikzpicture}
[xscale=1.25, yscale=1.5,
label distance=5pt,
every pin edge/.style={dotted,latex-,in=90,out=-45,shorten <=-4pt},
every pin/.style={dashed,pin distance=1cm}]

\draw[thick,rounded corners=6pt,densely dotted] (-.25,.25) -- (.25,0.25) -- (.25,-1.25) -- (-0.25,-1.25) -- cycle;
\draw[thick,rounded corners=6pt,densely dotted,xshift=1cm] (-0.25,0.25) -- (.25,0.25) -- (.25,-1.25) -- (-0.25,-1.25) -- cycle;

\draw[thick,rounded corners=6pt,densely dotted,xshift=-1.5cm] (4.25,0.25) -- (4.75,0.25) -- (5,-1.25) -- (4,-1.25) -- cycle;
\draw[thick,rounded corners=6pt,densely dotted] (4.25,0.25) -- (4.75,0.25) -- (5,-1.25) -- (4,-1.25) -- cycle;

\draw[thick,rounded corners=6pt,densely dotted] (6.75,0.25) -- (7.25,0.25) -- (7.8,-1.25) -- (6.2,-1.25) -- cycle;
\draw[thick,rounded corners=6pt,densely dotted,xshift=2cm] (6.75,0.25) -- (7.25,0.25) -- (7.8,-1.25) -- (6.2,-1.25) -- cycle;

\node[vertex] (v1) at (0,0) {};
\node[vertex] (v12) at (1,0) {};
\node (d1) at (2,-0.5) {$\dots$};
\node[vertex] (v21) at (3,0) {};
\node[vertex] (v22) at (4.5,0) {};
\node (d2) at (5.7,-0.5) {$\dots$};
\node[vertex] (v31) at (7,0) {};
\node[vertex] (v32) at (9,0) {};
\node (d3) at (10.25,-0.5) {$\dots$};
\node (d4) at (11.25,-0.5) {$\dots$};

\draw[decorate,decoration={brace,raise=12pt}] (v1.north west) -- node[above=12pt,sloped,fill=none,draw=none]{$\mbox{deg}_G(v)=1$} +(2.25,0);
\draw[decorate,decoration={brace,raise=12pt}] (v21.north west) -- node[above=12pt,sloped,fill=none,draw=none]{$\mbox{deg}_G(v)=2$} +(2.9,0);
\draw[decorate,decoration={brace,raise=12pt}] (v31.north west) -- node[above=12pt,sloped,fill=none,draw=none]{$\mbox{deg}_G(v)=3$} +(3.5,0);

\node[vertex2] (v11) at (0,-1) {};

\node[vertex2] (v12) at (1,-1) {};

\node[vertex2] (v211) at (2.75,-1) {};
\node[vertex2] (v212) at (3.25,-1) {};

\node[vertex2] (v221) at (4.25,-1) {};
\node[vertex2] (v222) at (4.75,-1) {};

\node[vertex2] (v311) at (6.5,-1) {};
\node[vertex2] (v312) at (7,-1) {};
\node[vertex2] (v313) at (7.5,-1) {};

\node[vertex2] (v321) at (8.5,-1) {};
\node[vertex2] (v322) at (9,-1) {};
\node[vertex2] (v323) at (9.5,-1) {};

\node[left=5pt of v1] (V) {$V$};
\node[left=5pt of v11] (L) {$L$};

\draw[rounded corners] (v11) -- +(0,-.4)  node (e) {} -|(v211);
\draw[rounded corners]  (v12)  -- +(0,-.5) -| (v311);
\draw[rounded corners]  (v221)  -- +(0,-.4) -| node (s) {}(v222);
\draw[rounded corners]  (v313)  -- +(0,-.4) -| node (m1) {}(v322);
\draw[rounded corners]  (v312)  -- +(0,-.5) -| node (m2) {}(v323);
\draw[rounded corners]  (v212) -- +(0,-.6) -| (v321);

\node[below right = .25 of e] (le) {edge};
\node[below left = .25 of m1] (lm) {multi-edge};
\node[below left = .25 of s] (ls) {self-loop};
\draw (le.west) edge[-latex,out=180,in=-135] (e.north);
\draw (lm.east) edge[-latex,out=0,in=-45] (m1.north);
\draw (lm.east) edge[-latex,out=0,in=-45] (m2.north);
\draw (ls.east) edge[-latex,out=0,in=-45] (s.north);

\end{tikzpicture}
 \label{fig:RandomMatch}
\end{figure}


 Given an $(\alpha,\beta)$ power law graph $G=(V,E)$ with $n$ vertices and maximum degree $\Delta$ and two integers $1\leq a\leq b\leq\Delta$, 
 an \emph{interval} $[a,b]$ is defined as the subset of $V$
 \begin{equation*}
 [a,b] = \{v\in V|a\leq \mbox{deg}_G(v)\leq b\}.
 \end{equation*}
 If $U\subseteq V$ is a subset of vertices, the \emph{volume} $\mbox{vol}(U)$ of $U$ is defined as the sum of node degrees of nodes in $U$.
 We will make use of estimates of sizes and volumes of node intervals in $(\alpha,\beta)$-PLGs.

\section{Approximation Schemes for \texorpdfstring{$\bm{0<\beta<2}$}{0<beta<2}}\label{sec:Approx}
We will now show that for every constant power law exponent $\beta\in (0,2)$, there is a PTAS for the MAX-CUT problem in $(\alpha,\beta )$-PLGs.
It turns out that for $\beta\in (0,1)$ this follows directly from the results in \cite{Arora1995}, since in that case the power law graphs are dense, 
(\autoref{section_4_1}). Recall that a graph $G=(V,E)$ with $n$ vertices 
is called $\delta$-dense if the number of edges satisfies $|E|\geq\delta\cdot n^2$. For $\beta =1$, $(\alpha,1)$-PLGs are not dense anymore. Nevertheless we can 
establish existence of a PTAS by showing that $(\alpha, 1)$-PLGs are \emph{core-dense}, a notion which was introduced in \cite{FernandezdelaVega2005}.

\subsection{The Case \texorpdfstring{$\bm{0<\beta\leq 1}$}{beta <= 1}}\label{section_4_1}
First we consider the case when the power law exponent $\beta$ is strictly less than $1$. In this case, the number $n$ of nodes
is asymptotically equal to $\frac{e^{\alpha\slash\beta}}{1-\beta}$ and the number $m$ of edges satisfies 
$m=\frac{1}{2}\frac{e^{2\alpha\slash\beta}}{2-\beta}$. Thus, in this case the graphs are asymptotically dense.
\begin{theorem}[\cite{Arora1995}]
For every $\delta>0$, there is a PTAS for MAX-CUT in $\delta$-dense graphs.
\end{theorem}
\begin{corollary}
There exists a PTAS for MAX-CUT in power law graphs with power law exponent $\beta<1$.
\end{corollary}
Now we consider the case when $\beta=1$. 
\begin{definition}\cite{FernandezdelaVega2005}
The \emph{core-strength} of a weighted $r$-uniform hypergraph $H=(V,E)$ with $|V|=n$ nodes given by an $r$-dimensional 
tensor $A\colon V\times\ldots\times V\to {\mathbb R}$ 
is 
\[ \left (\sum_{i=1}^nD_i \right )^{r-2}\sum_{i_1,\ldots , i_r\in V}\frac{A_{i_1,\ldots , i_r}^2}{\prod_{j=1}^r(D_{i_j}+\bar{D})},\]
where 
\[D_i=\sum_{i_2,\ldots , i_r\in V}A_{i,i_2,\ldots , i_r},\:\: \bar{D}=\frac{1}{n}\sum_{i=1}^nD_i\]
A class of weighted $r$-uniform hypergraphs is \emph{core-dense} if the core-strength is $O(1)$.
\end{definition}  
In particular, the class ${\mathcal C}$ of unweighted graphs is core-dense if 
\[\sup_{G=(V,E)\in {\mathcal C}}\sum_{\{i,j\}\in E}\frac{1}{(D_i+\bar{D})(D_j+\bar{D})}\: =\: O(1)\]
where $D_i$ is the degree of node $i$ in $G$ and $\bar{D}=\frac{1}{n}\sum_{i}D_i$ is the average-degree of $G$.
\begin{theorem}
For $\beta =1$, the class of $(\alpha,\beta )$-Power Law Graphs is core-dense.
\end{theorem}
\begin{corollary}
For $\beta =1$, there is a PTAS for MAX-CUT in $(\alpha, \beta)$-Power Law Graphs.
\end{corollary}
\emph{Proof of the Theorem.}
Let $G$ be an $(\alpha, 1)$-PLG. The average-degree of $G$ is asymptotically equal to
\[\bar{D}=\frac{1}{\alpha e^{\alpha}}\sum_{i=1}^{\Delta=e^{\alpha}}\left\lfloor\frac{e^{\alpha}}{i}\right\rfloor\cdot i
\geq\frac{1}{\alpha e^{\alpha}}\left (\sum_{i=1}^{e^{\alpha}}\frac{e^{\alpha}}{i}\cdot i - \sum_{i=1}^{e^{\alpha}} i\right )
 = \frac{1}{\alpha e^{\alpha}}\left (e^{\alpha\cdot 2}-\frac{e^{\alpha}(e^{\alpha}-1)}{2} \right ) = (1+o(1))\frac{e^{\alpha}}{2\alpha}\]
%
Now the core-strength of $G$ is 
\[\sum_{e=\{i,j\}\in E}\frac{1}{(\text{deg}(i)+\bar{D})(\text{deg}(j)+\bar{D})}
\leq \sum_{i=1}^{e^{\alpha}}\frac{e^{\alpha}}{i}\cdot\frac{1}{(i+\frac{e^{\alpha}}{2\alpha})^2}\leq\frac{\alpha^3}{e^{\alpha}}=o(1),\]
which concludes the proof of the theorem.\hfill$\Box$

\subsection{The Case \texorpdfstring{$\bm{1<\beta<2}$}{1<beta<2}}
We consider now the case when the power law exponent $\beta$ satisfies $1<\beta <2$.
Our approach is as follows. We choose a subset $[x\Delta ,\Delta ]$ of high-degree vertices and 
construct a cut for the subgraph $G_{[x\Delta, \Delta ]}$ induced by these vertices. Here $x\in (0,1)$ is a parameter of the construction.
We will show that we can choose $x$ in such a way that $G_{[x\Delta, \Delta ]}$ is dense and the volume of the residual set of vertices
$\left [1, x\Delta \right )$ is small, namely $\text{vol}(\left [1, x\Delta \right ))=o(|E|)$. 
Then we will construct a $(1+\epsilon )$-approximate solution for MAX-CUT on $G_{[x\Delta, \Delta ]}$ and afterwards place the remaining 
vertices arbitrarily. 

This approach is based on precise estimates for sizes and volumes of node degree intervals of the form $[1,x\Delta]$ and $[x\Delta,\Delta ]$, 
where $x\in (0,1)$ is the parameter of the construction and $\Delta=\lfloor e^{\nicefrac{\alpha}{\beta}}\rfloor$ is the maximum degree.
These estimates rely on the following lemma, which is also illustrated in \autoref{fig:StepFunction}.
\begin{lemma}
Let $f\colon {\mathbb R}^+\to {\mathbb R}^+$ be a monotone decreasing integrable convex function. Then we have the following bounds.
\begin{itemize}
\item[(a)] $\sum\limits_{a}^{b}\lfloor f(i)\rfloor\in\left [\int\limits_{a}^{b+1}f(t)dt -(b-a+1),\int\limits_{a}^{b+1}f(t)dt +f(a)-f(b+1)\right ]$
\item[(b)] $\sum\limits_{a}^{b}i\cdot \lfloor f(i)\rfloor\in
            \left [\int\limits_{a}^{b+1}tf(t)dt -\left (\frac{b(b+1)}{2}-\frac{(a-1)a}{2} \right ),
            \int\limits_{a}^{b+1}tf(t)dt +|(b+1)f(b+1)-af(a)| \right ]$
\end{itemize}
\end{lemma}

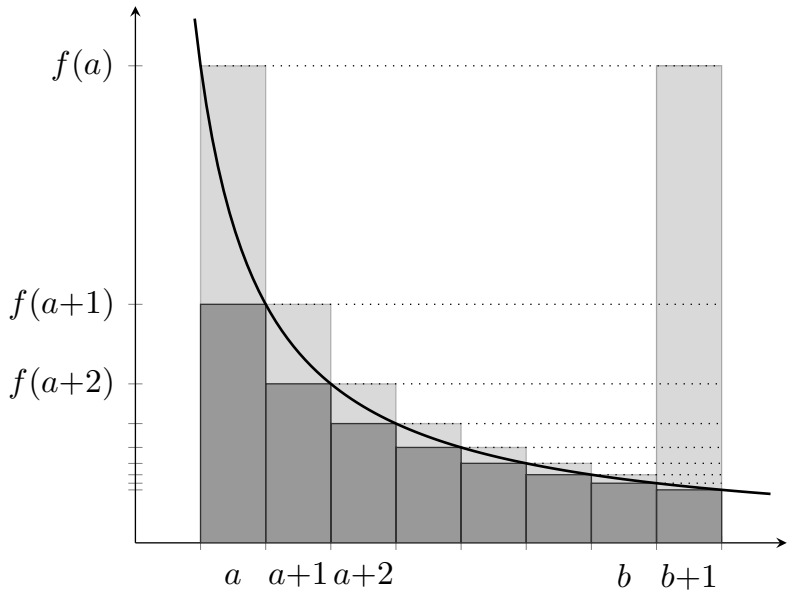
\begin{figure}[htb]
 \centering

\pgfplotsset{
    right segments/.code={\pgfmathsetmacro\rightsegments{#1}},
    right segments=3,
    right/.style args={#1:#2}{
        ybar interval,
        domain=#1+((#2-#1)/\rightsegments):#2+((#2-#1)/\rightsegments),
        samples=\rightsegments+1,
        x filter/.code=\pgfmathparse{\pgfmathresult-((#2-#1)/\rightsegments)}
    }
}

\pgfplotsset{
    left segments/.code={\pgfmathsetmacro\leftsegments{#1}},
    left segments=3,
    left/.style args={#1:#2}{
        ybar interval,
        domain=#1:#2,
        samples=\leftsegments+1,
        x filter/.code=\pgfmathparse{\pgfmathresult}
       }
}

\begin{tikzpicture}[/pgf/declare function={f=4/x;},scale=1.25]
\begin{axis}[
        xmin=0,xmax=10,ymin=0,ymax=4.5,
    domain=0.91:9.75,
    samples=100,
    axis lines=middle,
    ytick={0,4,2,1.333,1,.8,.666,.571,.5,.444},
    yticklabels={$0$,$f(a)$,$f(\medmuskip=0mu\relax a+1)$,$f(\medmuskip=0mu\relax a+2)$},
    xtick={0,1,...,9},
    xticklabels={$0$,$a\vphantom{1}$,$\medmuskip=0mu\relax a+1$,$\medmuskip=0mu\relax a+2$,,,,$b$,$\medmuskip=0mu\relax b+1$},
                 x tick label as interval,
    tick label style={font=\footnotesize},
]
\addplot [
    black!80,fill=gray,opacity=.3,
    left segments=7,
    left=1:8
] {f};

\addplot [
    black!80,fill=gray,opacity=.3,
    left segments=1,
    left=8:9
] {f(4)};

\addplot [
    black!80,fill=gray!80,opacity=1,
    right segments=8,
    right=1:9,
] {f};
\addplot [thick] {f};
\addplot [dotted,domain=1:9] {4};
\addplot [ dotted,domain=2:9] {2};
\addplot [ dotted,domain=3:9] {4/3};
\addplot [ dotted,domain=4:9] {1};
\addplot [ dotted,domain=5:9] {4/5};
\addplot [ dotted,domain=6:9] {4/6};
\addplot [ dotted,domain=7:9] {4/7};
\addplot [ dotted,domain=8:9] {4/8};
\end{axis}
\end{tikzpicture}
 \caption{Estimating the volumes of node degree intervals.}
 \label{fig:StepFunction}
\end{figure}

Using this lemma, we obtain the following bounds for the size of $[x\Delta ,\Delta ]$:
\begin{align*}
|[x\Delta, \Delta ]| &= \sum_{x\Delta}^{\Delta}\left\lfloor\frac{e^{\alpha}}{j^{\beta}}\right\rfloor\\
      &  \in \left [\int_{x\Delta}^{\Delta +1}\frac{e^{\alpha}}{t^{\beta}}dt-(\Delta (1-x)+1),
                     \int_{x\Delta}^{\Delta +1}\frac{e^{\alpha}}{t^{\beta}}dt +\frac{e^{\alpha}}{(x\Delta)^{\beta}}-\frac{e^{\alpha}}{(\Delta +1)^{\beta}}\right ]\\
      & = \left [ e^{\alpha}\left [\frac{t^{1-\beta}}{1-\beta} \right ]_{x\Delta}^{\Delta +1} -\Delta(1-x)-1,\:
                      e^{\alpha}\left [\frac{t^{1-\beta}}{1-\beta} \right ]_{x\Delta}^{\Delta +1} 
                      +\frac{e^{\alpha}}{(x\Delta)^{\beta}}-\frac{e^{\alpha}}{(\Delta +1)^{\beta}}\right ]\\
      & \subseteq \left [ \frac{e^{\alpha}}{\beta -1}\cdot\frac{1}{\Delta^{\beta -1}}\cdot\left (\frac{1}{x^{\beta -1}}-1 \right )
                            -\Delta(1-x)-1,\right.\\
      & \qquad\left.  \frac{e^{\alpha}}{\beta -1}\cdot\frac{1}{\Delta^{\beta -1}}\cdot\left (\frac{1}{x^{\beta -1}}-\frac{1}{2^{\beta -1}}\right )\:
                            + \frac{e^{\alpha}}{\Delta^{\beta }}\cdot \left (\frac{1}{x^{\beta }}-\frac{1}{2^{\beta }} \right ) \right ]\\
      & =  \left [ \Delta\left (\frac{x^{1-\beta}-1}{\beta -1}-1+x-\frac{1}{\Delta} \right ),\:
                            \frac{\Delta}{\beta -1}\left (\frac{1}{x^{\beta -1}}-\frac{1}{2^{\beta -1}}\right ) +\frac{1}{x^{\beta}}-\frac{1}{2^{\beta}}\right ] 
\end{align*}
We also obtain the following estimate for the volume of the interval $[1,x\Delta ]$.
\begin{align*}
  \text{vol}&([1,x\Delta]) = \sum_{1}^{x\Delta} j\cdot \left\lfloor\frac{e^{\alpha}}{j^{\beta}}\right\rfloor\\
      & \in \left [\int_{1}^{x\Delta +1}\frac{e^{\alpha}}{t^{\beta -1}}dt\: -\frac{x\Delta (x\Delta +1)}{2},\:
                     \int_{1}^{x\Delta +1}\frac{e^{\alpha}}{t^{\beta -1}}dt\: +e^{\alpha}\left (\frac{1}{(x\Delta -1)^{\beta -1}}-1\right ) \right ] \\
      & \subseteq  \left [ \frac{e^{\alpha}}{2-\beta }\left ((x\Delta +1)^{2-\beta }-1\right )-\frac{x\Delta (x\Delta +1)}{2},\right.\\
      &  \qquad \left. \frac{e^{\alpha}}{2-\beta }\left ((x\Delta +1)^{2-\beta }-1\right )
                                        +e^{\alpha}\left (\frac{1}{(x\Delta -1)^{\beta -1}}-1\right )\right ] \\
      & \subseteq  \left [(1-o(1))e^{\nicefrac{2\alpha}{\beta}}\left (\frac{x^{2-\beta}}{2-\beta }-\frac{x^2}{2} \right ),\:
                           (1+o(1))\frac{x^{2-\beta }}{2-\beta}e^{\nicefrac{2\alpha}{\beta}}\: +(1+o(1))\frac{e^{\nicefrac{\alpha}{\beta}}}{x^{\beta -1}}\right ]
\end{align*}
In particular, if $x$ is constant within the interval $(0,1)$, i.e. does not converge to $0$ or $1$, then we obtain the following estimates:
\begin{itemize}
\item The size of the interval $[x\Delta, \Delta ]$ is 
      $|[x\Delta , \Delta ]|=(1+o(1))\frac{e^{\alpha\slash\beta}}{\beta -1}\left (\frac{1}{x^{\beta -1}}-\frac{1}{2^{\beta -1}} \right )$.
\item Thus for $x$ being constant, the kernel function for the total error is $O(e^{2\alpha\slash\beta})$
\item For $x$ being constant, the volume of $[1,x\Delta ]$ is 
      $\text{vol}([1,x\Delta ])= (1+o(1)) \frac{x^{2-\beta}}{2-\beta}\cdot e^{2\alpha\slash\beta }$.
\end{itemize}
Thus we proceed as follows. Given an $(\alpha,\beta )$-PLG $G$ and $\epsilon >0$, we 
first choose $x\in (0,1)$ such that  
\[|E(G_{[x\Delta,\Delta ]})|\geq\frac{1}{2}\frac{e^{2\alpha\slash\beta}}{2-\beta}-\text{vol}([1,x\Delta ])\:
\geq\:\left (\frac{1}{2}-x^{2-\beta}\right )\frac{e^{2\alpha\slash\beta}}{2-\beta}\geq\left (1-\frac{1}{\tau (\epsilon )}\right )\cdot |E|,\]
where $\tau (\epsilon )$ is a function of $\epsilon$ yet to be defined.
Then we choose a second parameter $\epsilon'>0$ and construct a cut in $G_{[x\Delta,\Delta ]}$, using the 
AKK-algorithm with precision parameter $\epsilon$. This yields a cut in $G_{[x\Delta,\Delta ]}$ of size
at least $\text{MAX-CUT}(G_{[x\Delta,\Delta ]})-\epsilon'\cdot |[x\Delta,\Delta]|^2$. By the choice of $x$ we have
\[\text{MAX-CUT}(G_{[x\Delta,\Delta ]})\geq \text{MAX-CUT}(G)-\text{vol}([1,x\Delta])
\geq \text{MAX-CUT}(G)-\frac{1}{\tau (\epsilon )}\cdot |E|\]
Moreover,
\[|[x\Delta ,\Delta ]|^2\leq \frac{e^{2\alpha\slash\beta}}{(\beta -1)^2}\left (\frac{1}{x^{\beta -1}}-1\right )^2
=|E|\cdot \frac{2\cdot (2-\beta)}{(\beta -1)^2}\left (\frac{1}{x^{\beta -1}}-1\right )^2\]
Thus the size of the cut constructed in this way is at least
\begin{align*}&\text{MAX-CUT}(G)-\frac{|E|}{\tau (\epsilon )}
-\epsilon'\cdot |E|\cdot\frac{2(2-\beta )}{(\beta -1)^2}\left (\frac{1}{x^{\beta -1}}-1\right )^2 \\
&\geq \text{MAX-CUT}(G)
\left (1-\frac{2}{\tau (\epsilon )}-\epsilon'\cdot \frac{4(2-\beta )}{(\beta -1)^2}\left (\frac{1}{x^{\beta -1}}-1\right )^2\right ) 
\end{align*}
We want to achieve that this yields a $(1+\epsilon )$-approximation, i.e.
\begin{align*}
1-\frac{2}{\tau (\epsilon )}-\epsilon'\cdot \frac{4(2-\beta )}{(\beta -1)^2}\left (\frac{1}{x^{\beta -1}}-1\right )^2
&\geq\frac{1}{1+\epsilon}=1-\frac{\epsilon}{1+\epsilon},\\
\shortintertext{which is equivalent to} 
\frac{2}{\tau (\epsilon )}+\epsilon'\cdot \frac{4(2-\beta )}{(\beta -1)^2}\left (\frac{1}{x^{\beta -1}}-1\right )^2
&\leq \frac{\epsilon}{1+\epsilon}
\end{align*}
We achieve this in a two-step approach: First we define the function $\tau$ in such a way that
$\frac{2}{\tau (\epsilon )}\leq\frac{1}{2}\cdot\frac{\epsilon}{1+\epsilon}$, namely $\tau (\epsilon)=\frac{4(1+\epsilon )}{\epsilon}$.
Then we choose $\epsilon'$ appropriately, namely such that
\[\epsilon'\cdot \frac{4(2-\beta )}{(\beta -1)^2}\left (\frac{1}{x^{\beta -1}}-1\right )^2\leq\frac{1}{2}\cdot\frac{\epsilon}{1+\epsilon}
\text{,\quad or equivalently \quad} 
\epsilon'\leq
\frac{\frac{1}{2}\cdot\frac{\epsilon}{1+\epsilon}}
     {\frac{4(2-\beta)}{(\beta -1)^2}\left (2\tau (\epsilon )^{\frac{2-\beta}{\beta -1}}-1\right )^2}\]
Altogether we obtain the following result.
\begin{theorem}
For every fixed $1<\beta <2$, there is a PTAS for MAX-CUT in $(\alpha,\beta )$-Power Law Graphs.
\end{theorem}

Now we consider the case when the power law exponent is $\beta =2$.
In this case, the number of nodes is still linear in $e^{\alpha}$, but now the number
of edges drops down to $n\cdot\log (n)$. More precisely, the number of edges of an $(\alpha , 2)$-PLG is asymptotically equal to $\frac{1}{4}\alpha e^{\alpha}$, 
while the number of nodes is $\zeta (\beta )e^{\alpha}$.


\subsection{The Functional Case \texorpdfstring{$\bm{\beta_f=2-\frac{1}{f(\alpha)}}$}{}}\label{sec:functional}
We have shown in the previous sections that MAX-CUT on PLGs admits a PTAS for every fixed $\beta <2$. We consider now the 
functional case when $\beta_f=2-\nicefrac{1}{f(\alpha)}$, where $f(\alpha )$ is a monotone increasing function 
with $f(\alpha )\longrightarrow\infty$ as $\alpha\rightarrow\infty$. In this section we will show the following result.
\begin{theorem} 
For $\beta_f=2-\nicefrac{1}{f(\alpha )}$ with $f(\alpha )\longrightarrow\infty$ as $\alpha\rightarrow\infty$ and $f(\alpha)=o(\alpha)$,
there is a PTAS for MAX-CUT in $(\alpha,\beta_f)$-PLGs.
\end{theorem}
The proof of the Theorem is based on the following observation.
It is sufficient to show that we can split a given $(\alpha,\beta_f)$-PLG $G$ into two parts $\left [1,x\Delta_f\right )$ and 
$[x\Delta_f,\Delta_f]$ such that the following two conditions are satisfied:
\begin{itemize}
\item[(1)] $|[x\Delta_f, \Delta_f ]|^2=O(|E|)$,
\item[(2)] $\text{vol}[1,x\Delta_f,]=o(|E|)$.
\end{itemize}
Before we give the proof of the Theorem, we have to provide precise estimates for the sizes and volumes of these node degree intervals. 
The maximum degree is 
$\Delta_f=\lfloor e^{\alpha\slash\beta_f}\rfloor$. The number of vertices is 
\begin{align*}
\sum_{j=1}^{\Delta_f}\left\lfloor\frac{e^{\alpha}}{j^{\beta_f}}\right\rfloor 
 & \in \left [\sum_{j=1}^{\Delta_f}\frac{e^{\alpha}}{j^{\beta_f}}-\Delta_f,\: \sum_{j=1}^{\Delta_f}\frac{e^{\alpha}}{j^{\beta_f}} \right ]
\end{align*}
This sum can be approximated by the associated integral:
\begin{align*}
\sum_{j=1}^{\Delta_f}\frac{e^{\alpha}}{j^{\beta_f}}
 & \in \left [\int_{1}^{e^{\alpha\slash\beta_f}+1}\frac{e^{\alpha}}{z^{\beta_f}}dz,\: 
                \int_{1}^{e^{\alpha\slash\beta_f}+1}\frac{e^{\alpha}}{z^{\beta_f}}dz\: 
                +\frac{e^{\alpha}}{1^{\beta_f}}-\frac{e^{\alpha}}{(\lfloor e^{\alpha\slash\beta_f}\rfloor +1)^{\beta_f}} \right ] \\
 &\subseteq  \left [ e^{\alpha}\cdot\left [\frac{z^{-1+\nicefrac{1}{f(\alpha )}}}{-1+\nicefrac{1}{f(\alpha )}} \right ]_{1}^{e^{\alpha\slash\beta_f}+1},\:\:
                      e^{\alpha}\cdot\left [\frac{z^{-1+\nicefrac{1}{f(\alpha )}}}{-1+\nicefrac{1}{f(\alpha )}} \right ]_{1}^{e^{\alpha\slash\beta_f}+1}
                      \: +\: (1-o(1))e^{\alpha}\right ]\\
 &\subseteq  \left [ \frac{e^{\alpha}}{1-\nicefrac{1}{f(\alpha )}}\left (1-\frac{1}{e^{\alpha\cdot\frac{f(\alpha )-1}{2f(\alpha )-1}}}\right ),\:\:
                       \frac{e^{\alpha}}{1-\nicefrac{1}{f(\alpha )}}\left (1-\frac{1}{e^{\alpha\cdot\frac{f(\alpha )-1}{2f(\alpha )-1}}}\right )
                       \: +\: (1-o(1))e^{\alpha} \right ]
\end{align*}
Similarly we obtain:
\[\left [x\Delta_f,\Delta_f \right ]\: =\:\sum_{x\Delta_f}^{\Delta_f}\left\lfloor\frac{e^{\alpha}}{j^{\beta_f}}\right\rfloor
\:\in\:\left [\sum_{j=x\Delta_f}^{\Delta_f}\frac{e^{\alpha}}{j^{\beta_f}}-\Delta_f\left (1-x+\frac{1}{\Delta_f} \right ),\:\:
              \sum_{j=x\Delta_f}^{\Delta_f}\frac{e^{\alpha}}{j^{\beta_f}} \right ],\]
with 
\begin{align*}
\sum&_{j=x\Delta}^{\Delta_f}\frac{e^{\alpha}}{j^{\beta_f}}
  \in \left [\int_{xe^{\alpha\slash\beta_f}}^{e^{\alpha\slash\beta_f}+1}\frac{e^{\alpha}}{z^{\beta_f}}dz,\:
                \int_{xe^{\alpha\slash\beta_f}}^{e^{\alpha\slash\beta_f}+1}\frac{e^{\alpha}}{z^{\beta_f}}dz\:
                +\frac{e^{\alpha}}{(xe^{\alpha\slash\beta_f})^{\beta_f}}-\frac{e^{\alpha}}{(\lfloor e^{\alpha\slash\beta_f}\rfloor +1)^{\beta_f}} \right ]\\
 & \subseteq \left [\frac{e^{\alpha}}{1-\frac{1}{f(\alpha )}}
                      \left (\frac{1}{(xe^{\alpha\slash\beta_f})^{1-\frac{1}{f(\alpha )}}}
                            -(\frac{1}{(e^{\alpha\slash\beta_f}+1)^{1-\frac{1}{f(\alpha )}}} \right ),\right.\\
 & \qquad\left. \frac{e^{\alpha}}{1-\frac{1}{f(\alpha )}}
                      \left (\frac{1}{(xe^{\alpha\slash\beta_f})^{1-\frac{1}{f(\alpha )}}}
                            -\frac{1}{(e^{\alpha\slash\beta_f}+1)^{1-\frac{1}{f(\alpha )}}} \right )
                      +\frac{e^{\alpha}}{(xe^{\alpha\slash\beta_f})^{\beta_f}}
                      -\frac{e^{\alpha}}{(e^{\alpha\slash\beta_f}+1)^{\beta_f}}\right ]\\
 & \subseteq  \left [\frac{f(\alpha )}{f(\alpha )-1}\cdot e^{\alpha\cdot\frac{f(\alpha )}{2f(\alpha )-1}}
                      \cdot \left (\frac{1}{x^{\frac{f(\alpha )-1}{f(\alpha )}}}-1 \right ),\right.\\
 & \qquad \left. \frac{f(\alpha )}{f(\alpha )-1}\cdot e^{\alpha\cdot\frac{f(\alpha )}{2f(\alpha )-1}}
                      \cdot \left (\frac{1}{x^{\frac{f(\alpha )-1}{f(\alpha )}}}-\frac{1}{2^{\frac{2f(\alpha )-1}{f(\alpha )}}} \right ) 
                      \: +\frac{1}{x^{\frac{2f(\alpha )-1}{f(\alpha )}}}-\frac{1}{2^{\frac{2f(\alpha )-1}{f(\alpha )}}} \right ]
\end{align*}
Thus we obtain the following estimates for sizes of node degree intervals.
\begin{lemma}
Let $G=(V,E)$ be an $(\alpha,\beta_f)$-PLG with $\beta_f=2-\frac{1}{f(\alpha )}$. Then for every $0<x<1$, the size of the node degree interval 
$[x\Delta_f,\Delta_f]$ satisfies
\begin{align*}
|[x\Delta_f,\Delta_f]| & \in \left [\frac{f(\alpha )}{f(\alpha )-1}
                               \cdot e^{\alpha\cdot\frac{f(\alpha )}{2f(\alpha )-1}}
                               \cdot \left (\frac{1}{x^{\frac{f(\alpha )-1}{f(\alpha )}}}-1\right ) 
                               -\Delta_f\left (1-x+\frac{1}{\Delta_f} \right ), \right .\\
                       & \qquad   \left. \frac{f(\alpha )}{f(\alpha )-1}\cdot e^{\alpha\cdot\frac{f(\alpha )}{2f(\alpha )-1}}
                                      \cdot \left (\frac{1}{x^{\frac{f(\alpha )-1}{f(\alpha )}}}-\frac{1}{2^{\frac{2f(\alpha )-1}{f(\alpha )}}} \right) 
                                      \: +\frac{1}{x^{\frac{2f(\alpha )-1}{f(\alpha )}}}-\frac{1}{2^{\frac{2f(\alpha )-1}{f(\alpha )}}} \right ] 
\end{align*}
Moreover,
\begin{align*}
|V| & \in \left [ \frac{e^{\alpha}}{1-\nicefrac{1}{f(\alpha )}}\left (1-\frac{1}{e^{\alpha\cdot\frac{f(\alpha )-1}{2f(\alpha )-1}}}\right )-\Delta_f,\:
                       \frac{e^{\alpha}}{1-\nicefrac{1}{f(\alpha )}}\left (1-\frac{1}{e^{\alpha\cdot\frac{f(\alpha )-1}{2f(\alpha )-1}}}\right )
                        + (1-o(1))e^{\alpha} \right ]
\end{align*}
\end{lemma}
\begin{corollary}
The number of nodes of an $(\alpha,\beta_f)$-PLG $G=(V,E)$ with $\beta_f=2-\frac{1}{f(\alpha )}$ 
satisfies 
\[|V|\:\in\: \left [(1-o(1)\cdot \frac{f(\alpha )}{f(\alpha )-1}\cdot e^{\alpha},\: (1-o(1)\cdot \frac{2f(\alpha )-1}{f(\alpha)-1}\cdot e^{\alpha}\right ]\] 
\end{corollary}
Now we will estimate volumes of node degree intervals. Given some $x\in (0,1)$, possibly depending on $\alpha$, we have
\begin{align*}
\text{vol}&[x\Delta_f,\Delta_f ] \in \left [\int_{x\Delta_f}^{\Delta_f+1}\frac{e^{\alpha}}{z^{\frac{f(\alpha )-1}{f(\alpha )}}}dz
                                        -(1-x)\Delta_f^2,\right.\\
  &  \qquad\qquad\qquad \left. \int_{x\Delta_f}^{\Delta_f+1}\frac{e^{\alpha}}{z^{\frac{f(\alpha )-1}{f(\alpha )}}}dz\: 
                 + \frac{e^{\alpha}}{(x\Delta_f)^{\frac{f(\alpha )-1}{f(\alpha )}}}
                 - \frac{e^{\alpha}}{(\Delta_f+1)^{\frac{f(\alpha )-1}{f(\alpha )}}}\right ]\\
  & \subseteq \left [ e^{\alpha}\left [\frac{z^{1-\frac{f(\alpha )-1}{f(\alpha )}}}{1-\frac{f(\alpha )-1}{f(\alpha )}} \right ]_{x\Delta_f}^{\Delta_f+1}
                        -(1-x)e^{\frac{2\alpha\cdot f(\alpha )}{2f(\alpha )-1}},\right.\\
  &   \qquad \left. e^{\alpha}\left [\frac{z^{1-\frac{f(\alpha )-1}{f(\alpha )}}}{1-\frac{f(\alpha )-1}{f(\alpha )}} \right ]_{x\Delta_f}^{\Delta_f+1}
                 +\frac{e^{\alpha\cdot\frac{f(\alpha )-1}{2f(\alpha )-1}}}{x^{\frac{f(\alpha )-1}{f(\alpha )}}}
                 -\frac{e^{\alpha\cdot\frac{f(\alpha )-1}{2f(\alpha )-1}}}{2^{\frac{f(\alpha )-1}{f(\alpha )}}}\right ]\\
  & \subseteq \left [ f(\alpha )e^{\alpha\cdot\frac{2f(\alpha )}{2f(\alpha )-1}}
                      \left (\left (1+\frac{1}{\Delta_f} \right )^{\nicefrac{1}{f(\alpha )}}-x^{\nicefrac{1}{f(\alpha )}} \right )\:
                 - (1-x)e^{\frac{2\alpha\cdot f(\alpha )}{2f(\alpha )-1}},\right.\\
  &   \qquad \left.  f(\alpha )e^{\alpha\cdot\frac{2f(\alpha )}{2f(\alpha )-1}}
                      \left (\left (1+\frac{1}{\Delta_f} \right )^{\nicefrac{1}{f(\alpha )}}-x^{\nicefrac{1}{f(\alpha )}} \right )\:
                 + e^{\alpha\cdot\frac{f(\alpha )-1}{2f(\alpha )-1}}\cdot
                      \left (\frac{1}{x^{\frac{f(\alpha )-1}{f(\alpha )}}}-\frac{1}{2^{\frac{f(\alpha )-1}{f(\alpha )}}} \right ) \right ]
\end{align*}
Similarly we obtain
\begin{align*}
\text{vol}[1, x\Delta_f] & \in \left [ e^{\alpha}\left [\frac{z^{1-\frac{f(\alpha )-1}{f(\alpha )}}}{f(\alpha )^{-1}}\right ]_{1}^{x\Delta_f+1}
                                 -(x-\Delta_f^{-1})e^{\frac{2\alpha\cdot f(\alpha )}{2f(\alpha )-1}},\right.\\
       &   \qquad \left.  e^{\alpha}\left [\frac{z^{1-\frac{f(\alpha )-1}{f(\alpha )}}}{f(\alpha )^{-1}}\right ]_{1}^{x\Delta_f+1}
                     +1-\frac{e^{\alpha\cdot\frac{f(\alpha )-1}{2f(\alpha )-1}}}{x^{\frac{f(\alpha )-1}{f(\alpha )}}}\right ]\\
       & \subseteq \left [ f(\alpha )e^{\alpha}\left ( x^{\frac{1}{f(\alpha)}}e^{\frac{\alpha}{2f(\alpha )}}-1\right )
                     -xe^{\frac{2\alpha\cdot f(\alpha )}{2f(\alpha )-1}},\right.\\
       &   \qquad \left. f(\alpha )e^{\alpha}\left (2^{\nicefrac{1}{f(\alpha )}}x^{\frac{1}{f(\alpha)}}e^{\frac{\alpha}{2f(\alpha )}}-1\right )
                     +1-\frac{e^{\alpha\cdot\frac{f(\alpha )-1}{2f(\alpha )-1}}}{x^{\frac{f(\alpha )-1}{f(\alpha )}}}\right ]
\end{align*}
Thus the number of edges of an $(\alpha,\beta_f)$-PLG is 
\begin{align*}
|E| & \in \left [ f(\alpha )e^{\alpha}\left ((\Delta_f+1)^{\frac{1}{f(\alpha )}}-1\right )\: 
                    -\left (1-\frac{1}{\Delta_f}\right )e^{\alpha\cdot\frac{2f(\alpha )}{2f(\alpha )-1}},\right.\\
    &  \qquad \left.  f(\alpha )e^{\alpha}\left ((\Delta_f+1)^{\frac{1}{f(\alpha )}}-1\right )\:
                    +e^{\alpha\cdot\frac{f(\alpha )-1}{2f(\alpha )-1}}
                     \cdot\left (\Delta_f^{\frac{f(\alpha )-1}{f(\alpha )}}-\frac{1}{2^{\frac{f(\alpha )-1}{f(\alpha )}}} \right )\right ]\\
    & \subseteq \left [(f(\alpha )-1)e^{\alpha\cdot\frac{2f(\alpha )}{2f(\alpha )-1}},\:
                         f(\alpha )e^{\alpha\cdot\frac{2f(\alpha )}{2f(\alpha )-1}} +e^{\alpha\cdot\frac{2f(\alpha )-2}{2f(\alpha )-1}}\right ]\\
    & \subseteq \left [(f(\alpha )-1)e^{\alpha\cdot\frac{2f(\alpha )}{2f(\alpha )-1}},\:
                         (f(\alpha )+1)e^{\alpha\cdot\frac{2f(\alpha )}{2f(\alpha )-1}}\right ]
\end{align*}
\emph{Concerning (2):}\\
Based on these estimates, we will now show how to choose the parameter $x$ such as to satisfy both conditions (1) and (2).
It turns out that this depends on the order of growth of the function $f(\alpha )$,
We observe that condition (2) is equivalent to
\[x^{\nicefrac{1}{f(\alpha )}}e^{\alpha\cdot\frac{2f(\alpha )}{2f(\alpha )-1}}-e^{\alpha }\: =\: 
  o\left (e^{\alpha\cdot\frac{2f(\alpha )}{2f(\alpha )-1}}-e^{\alpha} \right )\]
We have
$e^{\alpha\cdot\frac{2f(\alpha )}{2f(\alpha )-1}}-e^{\alpha}\: 
=\: e^{\alpha}\cdot \left (e^{\alpha\cdot\frac{1}{2f(\alpha )-1}}-1\right )\:=\colon\: T_{f,\alpha}$.
Now we may consider three cases:
\begin{itemize}
\item[(a)] If $f(\alpha )=o(\alpha )$, then $T_{f,\alpha}=\omega (1)$,
\item[(b)] If $f(\alpha )=\Theta (\alpha )$, then $T_{f,\alpha}=\Theta \left (e^{\alpha }\right )$,
\item[(b)] If $f(\alpha )=\omega (\alpha )$, then $T_{f,\alpha}= o\left (e^{\alpha }\right )$.
\end{itemize}
We consider the case (a). Then, in order to satisfy the condition (2), we have to choose $x$ such that 
$x^{1\slash f(\alpha )} = o(1)$ and moreover, 
\[\frac{1}{e^{\alpha\cdot\frac{f(\alpha )}{2f(\alpha )-1}}}\:\:\leq\:\:  x\]
Now we consider the requirement (1). We observe that, up to constant factors, the condition in (1) is equivalent to
\[\left (\frac{f(\alpha )}{f(\alpha )-1}\right )^2e^{\alpha\cdot\frac{2f(\alpha )}{2f(\alpha )-1}}\cdot
\left (\frac{1}{x^{\frac{f(\alpha )-1}{f(\alpha )}}}-1 \right )^2\:\leq\:
f(\alpha )\cdot\left (e^{\alpha\cdot\frac{2f(\alpha )}{2f(\alpha )-1}}-e^{\alpha}\right ),\]
which is, by rearranging terms, equivalent to
\[\left (\frac{f(\alpha )}{f(\alpha )-1}\right )^2\cdot e^{\alpha}\cdot e^{\frac{\alpha}{2f(\alpha )-1}}\cdot\left (\frac{1}{x^{\frac{f(\alpha )-1}{f(\alpha )}}}-1 \right )^2
\:\leq\: f(\alpha )e^{\alpha}\left (e^{\frac{\alpha}{2f(\alpha )-1}}-1\right )\]
Since $\frac{f(\alpha )}{f(\alpha )-1}=\Theta (1)$, this is, again up to constant factors, equivalent to
\[\left (\frac{1}{x^{\frac{f(\alpha )-1}{f(\alpha )}}}-1 \right )^2\:\:\: \leq\:\:\: \frac{(f(\alpha)-1)^2}{f(\alpha )},\]
which yields the requirement
\[x\:\:\geq\:\: \frac{f(\alpha )^{\frac{f(\alpha )}{2(f(\alpha )-1)}}}{\left ((1+o(1))f(\alpha ) \right )^{\frac{f(\alpha )}{f(\alpha )-1}}}
\: =\: (1-o(1))\cdot f(\alpha )^{\frac{f(\alpha )}{2(f(\alpha )-1)}-\frac{f(\alpha )}{f(\alpha )-1}}=f(\alpha )^{-\frac{f(\alpha )}{f(\alpha )-1}}\]
Thus we obtain that for $x=\nicefrac{1}{f(\alpha )}$, requirements (1) and (2) are satisfied, which concludes the proof of the theorem. \hspace{\fill}$\Box$
%

\section{The Case \texorpdfstring{$\bm{\beta >2}$}{beta>2}}\label{sec:beta>2}
The situation drastically changes when the power law exponent is a constant $\beta >2$. 
In this case, the number of nodes is $|V|=\zeta (\beta )e^{\alpha}$, and the number of edges 
is $\frac{1}{2}\zeta (\beta -1)e^{\alpha}$. We apply the Goemans-Williamson algorithm to the graph $G_{[2,\Delta ]}$ induced by 
the vertices of degree at least $2$ in $G$, and afterwards place all the edges incident to degree-1 nodes in the cut.
This yields a cut which has an expected inverse approximation ratio at least 
\[\frac{\alpha_{GW}\cdot (\frac{1}{2}\zeta (\beta -1)-1)e^{\alpha}+\frac{1}{2}e^{\alpha}}{\frac{1}{2}\zeta (\beta -1)e^{\alpha}}=
\frac{\alpha_{GW}\cdot (\frac{1}{2}\zeta (\beta -1)-1)+\frac{1}{2}}{\frac{1}{2}\zeta (\beta -1)},\]
where $\alpha_{GW}\approx 0.879$ denotes the inverse approximation ratio of the Goemans-Williamson algorithm.
This analysis can be refined as follows. There are $e^{\alpha}$ nodes of degree $1$.
Suppose there are $\mu\cdot e^{\alpha}$ nodes of degree $1$ which are incident to another degree-1 node. The remaining 
$(1-\mu )e^{\alpha}$ degree-1 nodes are incident to nodes of higher degree. The resulting lower bound on the expected inverse
approximation ratio
is then
\[\frac{\alpha_{GW}\cdot (\frac{1}{2}\zeta (\beta-1)-(1-\frac{\mu}{2}))+(1-\frac{\mu}{2})}{\frac{1}{2}\zeta (\beta-1)}
\geq\frac{\alpha_{GW}\cdot (\frac{1}{2}\zeta (\beta -1)-\frac{1}{2})+\frac{1}{2}}{\frac{1}{2}\zeta (\beta -1)},\]
where the lower bound is attained at $\mu=1$.

\section{Approximation Lower Bounds for \texorpdfstring{\bm{$\beta>2$}}{beta>2}}\label{sec:LowerBounds}
In this section we provide explicit approximation lower bounds for MAX-CUT in $(\alpha,\beta)$-PLG for $\beta >2$.
Explicit lower bounds for the approximability of MAX-CUT in degree $B$ bounded graphs and $B$-regular graphs have been 
obtained in a series of papers \cite{Berman1999,Berman2001}.
The authors showed, among other results, that the MAX-CUT problem restricted to $3$-regular graphs is NP-hard to approximate to within $(\nicefrac{152}{151})+\epsilon$.

Here we will make use of this result and the associated constructions in order to prove APX-hardness for MAX-CUT in $(\alpha,\beta)$-PLG for $\beta>2$.
Moreover we will show that even for $\beta\in (0,2]$, the problem remains NP-hard in the exact setting. It turns out that for $\beta >1$ this will be a direct
consequence from our reduction for the case $\beta >2$, while in the case $\beta\leq 1$ a different construction is needed. This is due to the fact that 
for $\beta >1$, the number of constant degree nodes in an $(\alpha,\beta )$-PLG is linear in the total number of vertices, while for $\beta \leq 1$ this is not 
true anymore.

The section is organized as follows. In the next paragraph we will describe a generic construction which reduces the MAX-CUT problem in $3$-regular graphs to 
MAX-CUT in $(\alpha,\beta )$-PLGs, based on an embedding of the former graphs into the later ones. 
Afterwards we will use this construction such as to obtain the APX-hardness of MAX-CUT in PLGs for $\beta >2$, with an explicit approximation lower bound
that only depends on the power law exponent $\beta$.
This also yields the NP-hardness for $\beta\in (1,2]$.
Finally we will describe in subsection \ref{np_hardness_subsection} a different embedding construction which yields the NP-hardness in the exact setting for 
$\beta \leq 1$.


\paragraph{The Construction}

We will now describe an embedding of $3$-regular graphs into power law graphs.
Starting from an instance $G$ of E3-MAX-CUT, we construct an $(\alpha,\beta)$-PLG instance $G'=G \cup W\cup M$ which is the disjoint union of $G$, a matching $M$
on a subset of the degree-1 nodes and a subgraph $W=W_2\cup W_4\cup\ldots\cup W_{\Delta}$, where $\Delta =\lfloor e^{\nicefrac{\alpha}{\beta }}\rfloor$ is the 
maximum degree of $G$. For each $i\not\in\{1,3\}$, $W_i$ is a multipath of degree $i$ nodes, 
namely a path whose edges have multiplicities alternating between $\lfloor\frac{i}{2}\rfloor$ and $\lceil\frac{i}{2}\rceil$,
with degree $1$ nodes attached to the two endpoints of this path.
$W_i$ contains all the degree $i$ nodes of $G$ and additionally up to $i$ degree-1 nodes.

Given the graph $G$ with $N$ vertices, we choose $\alpha$ smallest possible such that $N$ is not larger than the number of degree $3$ vertices in 
an $(\alpha,\beta )$-PLG, namely such that 
\[N=\left\lfloor\frac{e^{\alpha}}{3^{\beta}}\right\rfloor\:\in\:\left [\frac{e^{\alpha}}{3^{\beta}}-1,\:\frac{e^{\alpha}}{3^{\beta}}\right ] \]
Without loss of generality, we assume that $N=\frac{e^{\alpha}}{3^{\beta}}=\left\lfloor\frac{e^{\alpha}}{3^{\beta}}\right\rfloor$.
In the construction of the subgraphs $W_i$, we have to distinguish the following four cases.\\[1ex]
\emph{Case 1: $i$ is even and $\lfloor\frac{e^{\alpha}}{i^{\beta}}\rfloor$ is even.}
Then $W_i$ contains vertices $v_{i,j},j=1,\ldots n_i=\lfloor\frac{e^{\alpha}}{i^{\beta}}\rfloor$ of degree $i$ and 
additionally $i$ vertices of degree $1$. For $1<j<n_i$, vertex $v_{i,j}$ is adjacent to $v_{i,j-1}$ and $v_{i,j+1}$ by multi-edges
of multiplicity $\nicefrac{i}{2}$. $v_{i,1}$ and $v_{i,n_i}$ are adjacent to $v_{i,2}$ and $v_{i,n_i-1}$ respectively, each by multi-edges with multiplicity
$\nicefrac{i}{2}$. Additionally, each of these two nodes is adjacent to $\nicefrac{i}{2}$ nodes of degree $1$.
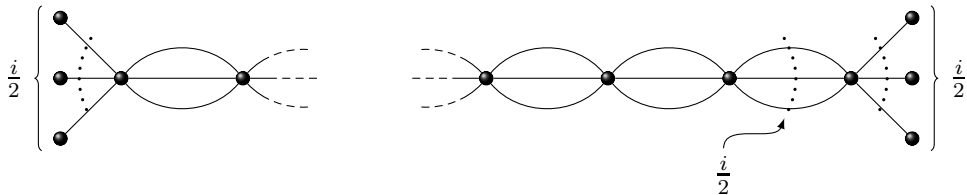
\begin{center}
 \begin{tikzpicture}[scale=1.6]

\foreach \x/\y in {0/a,1/b,3/c,4/d,5/e,6/f}{
 \node[vertex] (\y) at (\x,0) {};
}

\coordinate (b') at (2,0);
 \draw (c) edge[bend left=50,dash pattern=on 10pt off 2pt on 3pt off 2pt on 3pt off 2pt on 3pt off 100pt] (b');
 \draw (b) edge[bend right=50,dash pattern=on 10pt off 2pt on 3pt off 2pt on 3pt off 2pt on 3pt off 100pt] (b');
 \draw (c) edge[bend right=50,dash pattern=on 10pt off 2pt on 3pt off 2pt on 3pt off 2pt on 3pt off 100pt] (b');
 \draw (b) edge[bend left=50,dash pattern=on 10pt off 2pt on 3pt off 2pt on 3pt off 2pt on 3pt off 100pt] (b');
 \draw (b) edge[dash pattern=on 10pt off 2pt on 3pt off 2pt on 3pt off 2pt on 3pt off 100pt] (b');
 \draw (c) edge[dash pattern=on 10pt off 2pt on 3pt off 2pt on 3pt off 2pt on 3pt off 100pt] (b');

\draw (a) -- (b);
\draw (e) -- (f);
\draw (c) -- (d);
\draw (d) -- (e);

\foreach \x/\y in {a/b,c/d,d/e,e/f}{
 \draw (\x) edge[bend left=50] (\y);
 \draw (\x) edge[bend right=50] node [] (\x') {} (\y);
}

\foreach \x/\y in {.5/.5,.5/0,.5/-.5}{
 \draw (f) edge[] node[pos=1,vertex] {} +(\x,\y);
}
\foreach \x/\y in {.5/.5,.5/0,.5/-.5}{
 \draw (a) edge[] node[pos=1,vertex] {} +(-\x,-\y);
}


\draw[very thick,dots,xshift=5.45cm] +(0,.33) edge[bend left] node[pos=1](2){} +(0,-.33);
\draw[very thick,dots,xshift=6.2cm] +(0,.33) edge[bend left] +(0,-.33);
\draw[very thick,dots,xshift=-.25cm] +(0,.33) edge[bend right] +(0,-.33);

\draw[] ($(2)+(-.5,-.25)$) edge[-latex,in=-110,out=90] node[pos=-0,yshift=-10] {$\frac{i}{2}$}(2.center);

\draw[decorate,decoration={brace,raise=7pt}] ($ (f) + (.5,.6) $) -- node[right=10pt]{$\frac{i}{2}$} ($ (f) + (.5,-.6) $);
\draw[decorate,decoration={brace,raise=7pt,mirror}] ($ (a) + (-.5,.6) $) -- node[left=10pt]{$\frac{i}{2}$} ($ (a) + (-.5,-.6) $);

\end{tikzpicture}
\end{center}
\emph{Case 2: $i$ is even and $\lfloor\frac{e^{\alpha}}{i^{\beta}}\rfloor$ is odd.}
The construction is the same as in case 1.\\[1ex]
\emph{Case 3: $i$ is odd and $\lfloor\frac{e^{\alpha}}{i^{\beta}}\rfloor$ is even.}
Again $W_i$ has vertices $v_{i,j},j=1,\ldots n_i=\lfloor\frac{e^{\alpha}}{i^{\beta}}\rfloor$ of degree $i$ and
now additionally $i-1$ vertices of degree $1$. For $j=3,\ldots , n_i$, vertex $v_{i,j}$ is adjacent to $v_{i,j-1}$ with multiplicity $\lceil\frac{i}{2}\rceil$,
and for $j=1,3,\ldots n_i-2$ to $v_{i,j+1}$ with multiplicity $\lfloor\frac{i}{2}\rfloor$. Additionally $v_{i,1}$ has $\lfloor\frac{i}{2}\rfloor$ neighbors of degree $1$
and $v_{i,n_i}$ has $\lfloor\frac{i}{2}\rfloor$ degree $1$ neighbors. Thus $W_i$ contains $i-1$ nodes of degree $1$.
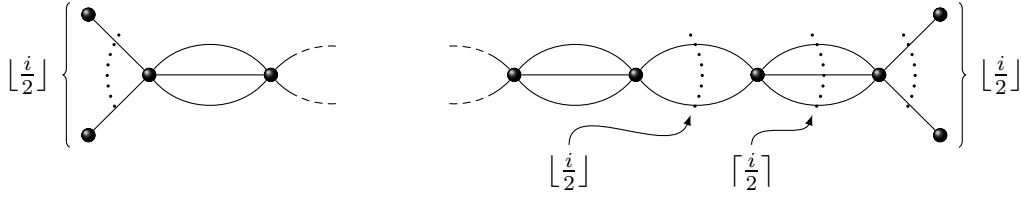
\begin{center}
 \begin{tikzpicture}[scale=1.6]

\foreach \x/\y in {0/a,1/b,3/c,4/d,5/e,6/f}{
 \node[vertex] (\y) at (\x,0) {};
}

\coordinate (b') at (2,0);
 \draw (c) edge[bend left=50,dash pattern=on 10pt off 2pt on 3pt off 2pt on 3pt off 2pt on 3pt off 100pt] (b');
 \draw (b) edge[bend right=50,dash pattern=on 10pt off 2pt on 3pt off 2pt on 3pt off 2pt on 3pt off 100pt] (b');
 \draw (c) edge[bend right=50,dash pattern=on 10pt off 2pt on 3pt off 2pt on 3pt off 2pt on 3pt off 100pt] (b');
 \draw (b) edge[bend left=50,dash pattern=on 10pt off 2pt on 3pt off 2pt on 3pt off 2pt on 3pt off 100pt] (b');

\draw (a) -- (b);
\draw (e) -- (f);
\draw (c) -- (d);

\foreach \x/\y in {a/b,c/d,d/e,e/f}{
 \draw (\x) edge[bend left=50] (\y);
 \draw (\x) edge[bend right=50] node [] (\x') {} (\y);
}


\foreach \x/\y in {.5/.5,.5/-.5}{
 \draw (a) edge[] node[pos=1,vertex] {} +(-\x,-\y);
}
\foreach \x/\y in {.5/.5,.5/-.5}{
 \draw (f) edge[] node[pos=1,vertex] {} +(\x,\y);
}

\draw[very thick,dots,xshift=4.45cm] +(0,.33) edge[bend left] node[pos=1](1){} +(0,-.33);
\draw[very thick,dots,xshift=5.45cm] +(0,.33) edge[bend left] node[pos=1](2){} +(0,-.33);
\draw[very thick,dots,xshift=6.2cm] +(0,.33) edge[bend left] +(0,-.33);
\draw[very thick,dots,xshift=-.25cm] +(0,.33) edge[bend right] +(0,-.33);

\draw[] ($(1)+(-1,-.25)$) edge[-latex,in=-110,out=90] node[pos=-0,yshift=-10] {$\lfloor\frac{i}{2}\rfloor$}(1.center);
\draw[] ($(2)+(-.5,-.25)$) edge[-latex,in=-110,out=90] node[pos=-0,yshift=-10] {$\lceil\frac{i}{2}\rceil$}(2.center);

\draw[decorate,decoration={brace,raise=7pt}] ($ (f) + (.5,.6) $) -- node[right=10pt]{$\lfloor\frac{i}{2}\rfloor$} ($ (f) + (.5,-.6) $);
\draw[decorate,decoration={brace,raise=7pt,mirror}] ($ (a) + (-.5,.6) $) -- node[left=10pt]{$\lfloor\frac{i}{2}\rfloor$} ($ (a) + (-.5,-.6) $);

\end{tikzpicture}
\end{center}
\emph{Case 4: $i$ is odd and $\lfloor\frac{e^{\alpha}}{i^{\beta}}\rfloor$ is odd.}
In this case, $v_{i,1}$ is adjacent to $\lceil\frac{i}{2}\rceil$ degree $1$ nodes and to $v_{i,2}$ with multiplicity $\lfloor\frac{i}{2}\rfloor$.
The vertex $v_{i,n_i}$ has $\lfloor\frac{i}{2}\rfloor$ degree $1$ neighbors and is adjacent to $v_{i,n_i-1}$ with multiplicity $\lceil\frac{i}{2}\rceil$.
Thus $W_i$ contains $\lfloor\frac{i}{2}\rfloor+\lceil\frac{i}{2}\rceil=i$ nodes of degree $1$.
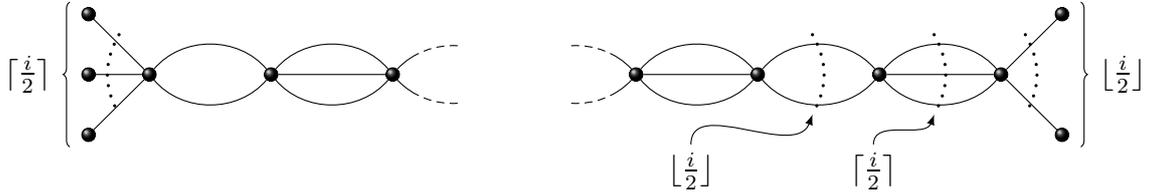
\begin{center}
 \begin{tikzpicture}[scale=1.6]

\foreach \x/\y in {-1/o,0/a,1/b,3/c,4/d,5/e,6/f}{
 \node[vertex] (\y) at (\x,0) {};
}

\coordinate (b') at (2,0);
 \draw (c) edge[bend left=50,dash pattern=on 10pt off 2pt on 3pt off 2pt on 3pt off 2pt on 3pt off 100pt] (b');
 \draw (b) edge[bend right=50,dash pattern=on 10pt off 2pt on 3pt off 2pt on 3pt off 2pt on 3pt off 100pt] (b');
 \draw (c) edge[bend right=50,dash pattern=on 10pt off 2pt on 3pt off 2pt on 3pt off 2pt on 3pt off 100pt] (b');
 \draw (b) edge[bend left=50,dash pattern=on 10pt off 2pt on 3pt off 2pt on 3pt off 2pt on 3pt off 100pt] (b');

\draw (a) -- (b);
\draw (e) -- (f);
\draw (c) -- (d);

\foreach \x/\y in {o/a,a/b,c/d,d/e,e/f}{
 \draw (\x) edge[bend left=50] (\y);
 \draw (\x) edge[bend right=50] node [] (\x') {} (\y);
}

\foreach \x/\y in {.5/.5,.5/0,.5/-.5}{
 \draw (o) edge[] node[pos=1,vertex] {} +(-\x,-\y);
}

\foreach \x/\y in {.5/.5,.5/-.5}{
 \draw (f) edge[] node[pos=1,vertex] {} +(\x,\y);
}

\draw[very thick,dots,xshift=4.45cm] +(0,.33) edge[bend left] node[pos=1](1){} +(0,-.33);
\draw[very thick,dots,xshift=5.45cm] +(0,.33) edge[bend left] node[pos=1](2){} +(0,-.33);
\draw[very thick,dots,xshift=6.2cm] +(0,.33) edge[bend left] +(0,-.33);
\draw[very thick,dots,xshift=-1.25cm] +(0,.33) edge[bend right] +(0,-.33);

\draw[] ($(1)+(-1,-.25)$) edge[-latex,in=-110,out=90] node[pos=-0,yshift=-10] {$\lfloor\frac{i}{2}\rfloor$}(1.center);
\draw[] ($(2)+(-.5,-.25)$) edge[-latex,in=-110,out=90] node[pos=-0,yshift=-10] {$\lceil\frac{i}{2}\rceil$}(2.center);

\draw[decorate,decoration={brace,raise=7pt}] ($ (f) + (.5,.6) $) -- node[right=10pt]{$\lfloor\frac{i}{2}\rfloor$} ($ (f) + (.5,-.6) $);
\draw[decorate,decoration={brace,raise=7pt,mirror}] ($ (o) + (-.5,.6) $) -- node[left=10pt]{$\lceil\frac{i}{2}\rceil$} ($ (o) + (-.5,-.6) $);

\end{tikzpicture}
\end{center}
Finally $M$ is a perfect matching on $\lfloor e^{\alpha}\rfloor -\sum_{i=4}^{\Delta}i+O(\Delta )=e^{\alpha}-\Theta (\Delta^2)=(1-o(1))e^{\alpha}$
vertices of degree $1$.

Our approximation lower bounds for MAX-CUT in $(\alpha , \beta )$-PLG will be based on the fact that a maximum cut in the subgraph $W$ contains
all the edges of $W$ and can be constructed efficiently.
\begin{lemma}\label{lower_bound_lemma}
Given a $3$-regular graph $G$ with $N$ vertices, let $W\cup M$ be the auxiliary subgraph described above. Then $W\cup M$ has the following properties.
\begin{itemize}
\item[(a)] The number of vertices is $|V(W\cup M)|\: =\: (1-o(1))\left (\zeta (\beta )-3^{-\beta}\right )e^{\alpha}$.
\item[(b)] The number of edges is $|E(W\cup M)|\: =\: (1-o(1))\left (\frac{1}{2}\zeta (\beta -1)-\frac{1}{2\cdot 3^{\beta -1}}\right )e^{\alpha}$.
\item[(c)] The cut size is $\mbox{MAX-CUT}(W\cup M)=|E(W\cup M)|$, and such a cut can be constructed in polynomial time.
\end{itemize}
\end{lemma}\begin{proof}
Parts (a) and (b) follow directly from the construction. Concerning (c), we place for each $W_i$ vertices $v_{i,j}$ on the left hand side of the cut if $j$ is odd and on
the right hand side if $j$ is even. Then we place all the degree $1$ nodes optimally. This concludes the proof of the lemma.
\end{proof}
\paragraph{Explicit Lower Bounds for \texorpdfstring{$\bm{\beta >2}$}{${\beta >2}$}}

We start from the following approximation hardness result for MAX-CUT in $3$-regular graphs.
\begin{theorem}[\cite{Berman2001}]\label{bk_theorem}
For every $\epsilon \in (0,\nicefrac{1}{302})$, it is NP-hard to decide whether an instance of E3-MAX-CUT with $156n$ edges and $104n$ vertices 
has a maximum cut of size above $(152-\epsilon)n$ or below $(151+\epsilon)n$. Thus the MAX-CUT problem in $3$-regular graphs is NP-hard to approximate within
any ratio $(\nicefrac{152}{151}) -\epsilon$. 
\end{theorem}
We consider now such a $3$-regular graph $G$ with $N=104n$ vertices and construct the 
associated $(\alpha,\beta )$-PLG $G'=G\cup W\cup M$. Due to \autoref{lower_bound_lemma}, the graph $G'$ contains 
$(1-o(1))\left (\zeta (\beta )-3^{-\beta}\right )\cdot 3^{\beta}\cdot N+N$ vertices, and the auxiliary subgraph $W\cup M$ contains 
$(1-o(1))\left (\frac{1}{2}\zeta (\beta -1)-\frac{1}{2\cdot 3^{\beta -1}}\right )\cdot 3^{\beta}\cdot N$ edges. It is now NP-hard to decide if 
\[\mbox{MAX-CUT}(G')\geq (1-o(1))\left (\frac{1}{2}\zeta (\beta -1)-\frac{1}{2\cdot 3^{\beta -1}}\right )\cdot 3^{\beta}\cdot 104n +(152-\epsilon )n\]
or
\[\mbox{MAX-CUT}(G')\leq (1-o(1))\left (\frac{1}{2}\zeta (\beta -1)-\frac{1}{2\cdot 3^{\beta -1}}\right )\cdot 3^{\beta}\cdot 104n +(151+\epsilon )n\]  
Simplifying terms, we obtain the following result.
\begin{theorem}
For every $\beta >2$
and $\epsilon >0$,
MAX-CUT in $(\alpha,\beta )$-Power Law Graphs is NP-hard to approximate within a ratio
$\frac{(3^{\beta} \zeta (\beta -1)-3)\cdot 52+152}{(3^{\beta} \zeta (\beta -1)-3)\cdot 52+151}-\epsilon$.
\end{theorem}
As a byproduct of the proof of the previous theorem we also obtain the following result.
\begin{corollary}
For $1<\beta \leq 2$, MAX-CUT in $(\alpha,\beta )$-Power Law Graphs is NP-hard.
\end{corollary}
\begin{proof}
Theorem \ref{bk_theorem} yields that the following decision problem is NP-complete: Given a $3$-regular graph $G$ with $104n$ vertices, is 
$\mbox{MAX-CUT}(G)>(152-\epsilon)n$? For $\beta >1$, our reduction $G\mapsto G'=G\cup W\cup M$ is well defined and reduces this to the decision problem 
if $\mbox{MAX-CUT}(G')>(152-\epsilon )n+\frac{1}{2}\sum_{j=1}^{\Delta}j\cdot\lfloor\frac{e^{\alpha}}{j^{\beta}}\rfloor-156n$, where $e^{\alpha}=3^{\beta}\cdot 104n$.
\end{proof}
{\bf Remark.} The reduction is not well-defined anymore for $\beta\leq 1$, since in that case the number of degree $1$ nodes does not suffice to construct the 
subgraphs $W_i,i=2,4,5,\ldots ,\Delta$ in the way as described before. In the next subsection we will provide an alternative reduction which also 
yields the NP-hardness for $\beta \leq 1$.

\subsection{NP-Hardness for \texorpdfstring{$\bm{\beta \leq 1}$}{beta <= 1}}\label{np_hardness_subsection}
In order to prove NP-hardness of the MAX-CUT problem in $(\alpha,\beta )$-power law graphs for $\beta \leq 1$, we construct again a polynomial time reduction from
the $3$-regular MAX-CUT. We consider first the case $\beta <1$. Then the number of degree $1$ nodes is still equal to $\lfloor e^{\alpha}\rfloor$, while the 
total number of nodes is $(1-o(1))\frac{e^{\nicefrac{\alpha}{\beta}}}{1-\beta}=\omega (e^{\alpha})$. Thus it is even not possible to spend one degree $1$ node 
per subgraph $W_i$, since we also have $e^{\alpha}=o(\Delta )$. 

We propose the following alternative construction. Starting from a $3$-regular graph $G$ with $N=104n$ vertices, we choose again $\alpha$ such that 
$N=\lfloor\frac{e^{\alpha}}{3^{\beta}}\rfloor=\frac{e^{\alpha}}{3^{\beta}}$. Now we call an integer $i\in\{4,\ldots\Delta\}$ {\sl critical} if 
both $i$ and $\lfloor\frac{e^{\alpha}}{i^{\beta}}\rfloor$ are odd. 
 
For those node degrees $i\in\{2,4,5,\ldots,\Delta\}$ which are {\sl non-critical}, we let $W_i$ be a wheel, i.e. a cycle consisting of multi-edges containing
all the degree $i$ nodes from $G'$ and having edge multiplicities alternating between $\lfloor\frac{i}{2}\rfloor$ and $\lceil\frac{i}{2}\rceil$.
In that case we have 
\[\mbox{MAX-CUT}(W_i)\:\: =\:\left\{\begin{array}{l@{\:\:}l}
 \frac{1}{2}\cdot i\cdot n_i & \mbox{if $i$ is even,}\\[0.667ex]
 \frac{n_i\cdot i}{2}-\left\lfloor\frac{i}{2}\right\rfloor & \mbox{if $i$ is odd and $n_i$ even.}
\end{array}\right.\] 
Now we deal with the {\sl critical node degrees}. Suppose that $i_1<i_2<\ldots <i_c$ are the critical degrees in $\{2,4,5,\ldots ,\Delta\}$.
We take a maximum matching $M_c$ on this set of indices such that without loss of generality, index $i_c$ is unmatched iff $c$ is odd. For each pair $i,j$ in this matching, 
we construct a subgraph $W_{i,j}$ containing $n_i$ nodes of degree $i$ and $n_j$ nodes of degree $j$. This subgraph is constructed as follows. We 
arrange all the degree $i$ nodes on a cycle consisting of multi-edges with multiplicities alternating between $\lfloor\frac{i}{2}\rfloor$ and $\lceil\frac{i}{2}\rceil$.
Let $v_i$ be the node which has now degree $i-1=2\cdot\lfloor\frac{i}{2}\rfloor$. All the other nodes on this cycle have already degree $i$.
Now we do the same for the $n_j$ nodes which are supposed to have degree $j$, and we define the vertex $v_j$ accordingly. 
Finally we add a single edge of multiplicity $1$ connecting $v_i$ and $v_j$.
In \autoref{fig:CyclicHardness} we show the construction of $W_{i,j}$.

\begin{figure}[htb]
 \centering
 \begin{tikzpicture}[scale=1.6]

\draw (-1,1) arc (90:270:.5);
\draw (-1,1) edge[dash pattern=on 10pt off 2pt on 3pt off 2pt on 3pt off 2pt on 3pt off 100pt] +(1,0);
\draw (-1,.75) arc (90:270:.375);
\draw (-1,.75) edge[dash pattern=on 10pt off 2pt on 3pt off 2pt on 3pt off 2pt on 3pt off 100pt] +(1,0);
\draw (-1,.5) arc (90:270:.25);
\draw (-1,.5) edge[dash pattern=on 10pt off 2pt on 3pt off 2pt on 3pt off 2pt on 3pt off 100pt] +(1,0);

\draw (6,1) arc (90:-90:.5);
\draw (6,1) edge[dash pattern=on 10pt off 2pt on 3pt off 2pt on 3pt off 2pt on 3pt off 100pt] +(-1,0);
\draw (6,.75) arc (90:-90:.375);
\draw (6,.75) edge[dash pattern=on 10pt off 2pt on 3pt off 2pt on 3pt off 2pt on 3pt off 100pt] +(-1,0);
\draw (6,.5) arc (90:-90:.25);
\draw (6,.5) edge[dash pattern=on 10pt off 2pt on 3pt off 2pt on 3pt off 2pt on 3pt off 100pt] +(-1,0);

\foreach \x/\y in {-1/o,0/a,1/b,3/c,4/d,5/e,6/f}{
 \node[vertex] (\y) at (\x,0) {};
}

\coordinate (m) at (2,0);
 \draw (c) edge[bend left=50,dash pattern=on 10pt off 2pt on 3pt off 2pt on 3pt off 2pt on 3pt off 100pt] (m);
 \draw (b) edge[bend right=50,dash pattern=on 10pt off 2pt on 3pt off 2pt on 3pt off 2pt on 3pt off 100pt] (m);
 \draw (c) edge[bend right=50,dash pattern=on 10pt off 2pt on 3pt off 2pt on 3pt off 2pt on 3pt off 100pt] (m);
 \draw (b) edge[bend left=50,dash pattern=on 10pt off 2pt on 3pt off 2pt on 3pt off 2pt on 3pt off 100pt] (m);
 \draw (b) edge[dash pattern=on 10pt off 2pt on 3pt off 2pt on 3pt off 2pt on 3pt off 100pt] (m);
 \draw (c) edge[dash pattern=on 10pt off 2pt on 3pt off 2pt on 3pt off 2pt on 3pt off 100pt] (m);

\draw (a) -- (o);
\draw (d) -- (e);

\foreach \x/\y in {o/a,a/b,c/d,d/e,e/f}{
 \draw (\x) edge[bend left=50] (\y);
 \draw (\x) edge[bend right=50]  (\y);
}
\foreach \x/\y in {a/b,c/d,e/f}{
 \draw (\x) edge[bend left=20] (\y);
 \draw (\x) edge[bend right=20]  (\y);
}

\draw[very thick,dots,xshift=-1.3cm,yshift=.2cm,rotate=-70] +(0,.33) edge[bend right] node[pos=1](3){} +(0,-.33);
\draw[very thick,dots,xshift=-.45cm] +(0,.33) edge[bend right] node[pos=1](1){} +(0,-.33);
\draw[very thick,dots,xshift=6.3cm,yshift=.2cm,rotate=70] +(0,.33) edge[bend left] node[pos=1](4){} +(0,-.33);
\draw[very thick,dots,xshift=5.45cm] +(0,.33) edge[bend left] node[pos=1](2){} +(0,-.33);

\draw[] ($(1)+(1,-.25)$) edge[-latex,in=-70,out=90] node[pos=-0,yshift=-10] {$\lfloor\frac{i}{2}\rfloor$}(1.center);
\draw[] ($(2)+(-1,-.25)$) edge[-latex,in=-110,out=90] node[pos=-0,yshift=-10] {$\lceil\frac{i}{2}\rceil$}(2.center);
\draw[] ($(3)+(0,-.66)$) edge[-latex,in=-110,out=90] node[pos=-0,yshift=-10] {$\lfloor\frac{i}{2}\rfloor$}(3.center);
\draw[] ($(4)+(0,-.66)$) edge[-latex,in=-70,out=90] node[pos=-0,yshift=-10] {$\lfloor\frac{i}{2}\rfloor$}(4.center);

\draw[] ($(o)+(0.4,-.65)$) edge[-latex,in=-90,out=90] node[draw,rounded corners,below,pos=-0] {\footnotesize\emph{critical} node}(o);

\begin{scope}[xshift=.5cm,yshift=2.5cm]

\draw (-1,0) arc (90:270:.375);
\draw (-1,-.75) edge[dash pattern=on 10pt off 2pt on 3pt off 2pt on 3pt off 2pt on 3pt off 100pt] +(1,0);
\draw (-1,0) arc (90:270:.25);
\draw (-1,-.5) edge[dash pattern=on 10pt off 2pt on 3pt off 2pt on 3pt off 2pt on 3pt off 100pt] +(1,0);

\draw (6,0) arc (90:-90:.375);
\draw (6,-.75) edge[dash pattern=on 10pt off 2pt on 3pt off 2pt on 3pt off 2pt on 3pt off 100pt] +(-1,0);
\draw (6,0) arc (90:-90:.25);
\draw (6,-.5) edge[dash pattern=on 10pt off 2pt on 3pt off 2pt on 3pt off 2pt on 3pt off 100pt] +(-1,0);

\foreach \x/\y in {-1/o',0/a',1/b',2/c',4/d',5/e',6/f'}{
 \node[vertex] (\y) at (\x,0) {};
}

\draw (o') -- (a');
\draw (b') -- (c');
\draw (d') -- (e');

\foreach \x/\y in {o'/a',a'/b',b'/c',d'/e',e'/f'}{
 \draw (\x) edge[bend left=50] (\y);
 \draw (\x) edge[bend right=50] (\y);
}

\coordinate (m') at (3,0);
 \draw (c') edge[bend left=50,dash pattern=on 10pt off 2pt on 3pt off 2pt on 3pt off 2pt on 3pt off 100pt] (m');
 \draw (d') edge[bend right=50,dash pattern=on 10pt off 2pt on 3pt off 2pt on 3pt off 2pt on 3pt off 100pt] (m');
 \draw (c') edge[bend right=50,dash pattern=on 10pt off 2pt on 3pt off 2pt on 3pt off 2pt on 3pt off 100pt] (m');
 \draw (d') edge[bend left=50,dash pattern=on 10pt off 2pt on 3pt off 2pt on 3pt off 2pt on 3pt off 100pt] (m');

\draw[very thick,dots,xshift=-1.3cm,yshift=-.2cm,rotate=70] +(0,.33) edge[bend right] node[pos=0](3'){} +(0,-.33);
\draw[very thick,dots,xshift=-.45cm] +(0,.33) edge[bend right] node[pos=0](1'){} +(0,-.33);
\draw[very thick,dots,xshift=6.3cm,yshift=-.2cm,rotate=-70] +(0,.33) edge[bend left] node[pos=0](4'){} +(0,-.33);
\draw[very thick,dots,xshift=5.45cm] +(0,.33) edge[bend left] node[pos=0](2'){} +(0,-.33);

\draw[] ($(1')+(1,.25)$) edge[-latex,in=70,out=-90] node[above=12pt,pos=-0,yshift=-10] {$\lceil\frac{j}{2}\rceil$}(1'.center);
\draw[] ($(2')+(-1,.25)$) edge[-latex,in=110,out=-90] node[above=12pt,pos=-0,yshift=-10] {$\lfloor\frac{j}{2}\rfloor$}(2'.center);
\draw[] ($(3')+(0,.66)$) edge[-latex,in=110,out=-90] node[above=12pt,pos=-0,yshift=-10] {$\lfloor\frac{j}{2}\rfloor$}(3'.center);
\draw[] ($(4')+(0,.66)$) edge[-latex,in=70,out=-90] node[above=12pt,pos=-0,yshift=-10] {$\lfloor\frac{j}{2}\rfloor$}(4'.center);

\draw[] ($(f')+(-0.4,.7)$) edge[-latex,in=90,out=-90] node[draw,rounded corners,above=0pt,pos=-0] {\footnotesize\emph{critical} node}(f');

\draw[decorate,decoration={brace,mirror}] (-2,.25) -- node[left]{$W_{j}$} (-2,-1);
\draw[densely dotted] (-2,-1) -- (-1.5,-1);
\draw[densely dotted] (-2,.25) -- (-1.5,.25);

\end{scope}

\draw[very thick,rounded corners] (o) -- +(0,1.375) -| node[fill=white,pos=0.25]{\footnotesize edge of multiplicity $1$} (f');
\draw[decorate,decoration={brace}] (-2,-.25) -- node[left]{$W_{i}$} (-2,1.25);
\draw[densely dotted] (-2,1.25) -- (-1.5,1.25);
\draw[densely dotted] (-2,-.25) -- (-1.5,-.25);
\end{tikzpicture}
 \caption{The construction of $W_{i,j}$.}
 \label{fig:CyclicHardness}
\end{figure}
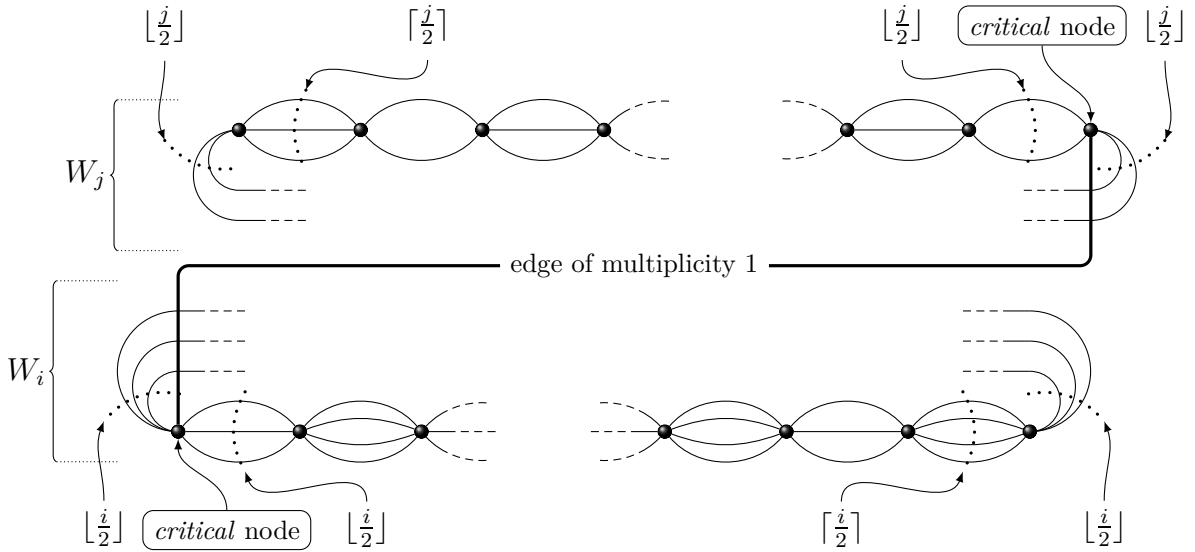

Directly from this construction we obtain that the subgraph $W_{i,j}$ contains $n_i$ nodes of degree $i$ and $n_j$ nodes of degree $j$. Note that from the very 
definition of critical degrees it follows that both $n_i$ and $n_j$ are odd. Every cut in $W_{i,j}$ misses at least one multi-edge inside the degree $i$ nodes and one 
such edge inside the degree $j$ nodes. On the other hand, we can easily construct such a cut which misses only one edge of multiplicity $\lfloor\frac{i}{2}\rfloor$
and one edge of multiplicity $\lfloor\frac{j}{2}\rfloor$, by placing the degree $i$ nodes alternatingly on the left and right hand side of the cut and doing the same 
for the degree $j$ vertices. Thus we have
\[\mbox{MAX-CUT}(W_{i,j})\:\: =\:\: \frac{i\cdot n_i}{2}-\left\lfloor\frac{i}{2}\right\rfloor
                                 \: +  \frac{j\cdot n_j}{2}-\left\lfloor\frac{j}{2}\right\rfloor\: +1\]
Finally if $c$ is odd, then $W_{i_c}$ contains all the degree-$i_c$ nodes and one node of degree $1$, and the maximum cut size is
$\mbox{MAX-CUT}(W_{i_c})= \frac{i_c\cdot n_{i_c}}{2}-\lfloor\frac{i}{2}\rfloor +1$. 

In order to keep notation simple, we let $W_{i_c}=\emptyset$ in case when the number of {\sl critical} node degrees is even. Moreover we let 
$J\subset\{2,4,\ldots , \Delta\}$ be the set of critical degrees.
Thus the map $G\mapsto G'=G\cup\bigcup_{i\not\in J}W_i\cup\bigcup_{\{i,j\}\in M_c}W_{i,j}\cup W_{i_c}\cup M$ 
reduces the question if a given $3$-regular graph $G$ with $104n$ vertices has a 
cut of size at least $(152-\epsilon )n$ to the question if the $(\alpha,\beta )$-power law graph $G'$ has a cut of size at least\\[1.1ex]
$(152-\epsilon )n+\sum\limits_{i\in ([2,\Delta]\setminus J)\cup\{i_c\}}\mbox{MAX-CUT}(W_i)
\: +\sum\limits_{\{i,j\}\in M_c}\mbox{MAX-CUT}(W_{i,j})\: +\mbox{MAX-CUT}(M)$\\[1.1ex]
The same construction also works in the case when $\beta =1$.
We obtain the following result.
\begin{theorem}
For every $\beta\in \left (0,1\right ]$, the MAX-CUT problem in $(\alpha,\beta )$-Power Law Graphs is NP-hard.
\end{theorem}


\section{Further Research}
We prove some new results on approximability of MAX-CUT 
in Power Law Graphs. 
It remains an open problem to settle the status of MAX-CUT in PLGs for the power law exponent $\beta =2$, where the phase transition happens from 
existence of a PTAS (for constant $\beta <2$ and for functional $\beta$ slowly converging to $2$ from below) to APX-hardness (for any constant $\beta >2$). 
Another problem concerns the design of better constant factor approximation algorithms for MAX-CUT in the case $\beta >2$, based on SDP simulations with 
appropriate classes of inequalities for the low-degree vertices in the power law graph.
Similar methods can be applied to other partition problems on power law graphs like MAX-BISECTION, Multiway-CUT and $k$-partition problems.
Establishing good approximability bounds for those problems is another interesting question.

\printbibliography


%

\end{document}